\documentclass[11pt]{article}%
\usepackage{amsmath}
\usepackage{graphics}
\usepackage{color}
\usepackage{epsfig}
\usepackage{graphicx}%
\usepackage{amsfonts}%
\usepackage{amssymb}
\usepackage{amsthm}

\usepackage{setspace}
\usepackage[margin=0.8in]{geometry}
\usepackage{comment}
\newtheorem{theorem}{Theorem}[section]

\newtheorem{definition}{Definition}[section]
\newtheorem{claim}{Claim}[section]
\newtheorem{lemma}{Lemma}[section]

\newtheorem{corollary}{Corollary}[section]

\newtheorem{observation}{Observation}[section]

\newtheorem{example}{Example}[section]

\newenvironment{oneshot}[1]{\@beginlemma{#1}{\unskip}}{\@endlemma}

\newcommand{\lbfl}{{\sc Lbfl}}
\newcommand{\cfl}{{\sc Cfl}}
\newcommand{\ufl}{{\sc Ufl}}

\date{}

\allowdisplaybreaks

\begin{document}

\title{
Integrality gaps for strengthened LP relaxations of Capacitated
and Lower-Bounded Facility Location\thanks{
This research has been co-financed by the European Union (European
Social Fund -- ESF) and Greek national funds through the Operational
Program ``Education and Lifelong Learning'' of the National Strategic
Reference Framework (NSRF) - Research Funding Program:
``Thalis. Investing in knowledge society through the European Social Fund''.}
}

\author {Stavros G. Kolliopoulos\thanks{Department of Informatics and
Telecommunications, National and Kapodistrian 
University of Athens, Panepistimiopolis Ilissia, Athens
157 84, Greece; (\texttt{www.di.uoa.gr/}\~{\tt sgk}). Part of this
work conducted while visiting the 
  IEOR Department, Columbia University, New York, NY 10027.}   
\and Yannis Moysoglou\thanks{ 
Department of Informatics and
Telecommunications, National and Kapodistrian 
University of Athens, Panepistimiopolis Ilissia, Athens
157 84, Greece; (\texttt{gmoys@di.uoa.gr}). Partially supported by an
NKUA-ELKE graduate fellowship.} }

\date{July 8, 2013}

\maketitle

\thispagestyle{empty}

\begin{abstract}
The metric uncapacitated facility location problem (\ufl)  enjoys a
special stature in approximation algorithms 
as  a testbed for various techniques, among which  LP-based
methods  have been especially prominent and successful.  
Two generalizations of \ufl\  are  {\em capacitated facility
location (\cfl\/)} and 
{\em lower-bounded facility location (\lbfl\/).}
In the former, every facility has a capacity which is the maximum demand that can be assigned
to it, while in the latter, every open facility is required to serve a given minimum amount of
demand.
Both \cfl\ and \lbfl\ are 
approximable within a constant factor but  
their respective  natural LP relaxations
have an unbounded  integrality gap. 
One could hope that different, less natural relaxations might provide 
better lower bounds.  According to Shmoys and Williamson, the existence of a
relaxation-based algorithm for \cfl\ is one of the top 10  open
problems  in approximation algorithms. 

In this paper we give the first results on this problem and they are
negative in nature.  We  show  unbounded
integrality gaps for two substantial families of strengthened formulations.

The first family we consider  is the 
hierarchy of LPs resulting from 
repeated applications of the lift-and-project Lov\'{a}sz-Schrijver
procedure starting from the standard relaxation. We show that
 the LP
relaxation for \cfl\ resulting after $\Omega(n)$ rounds, where $n$ is
the number of facilities in the instance, has  unbounded integrality
gap. Note that  the Lov\'{a}sz-Schrijver
procedure is known to yield an exact formulation for \cfl\ in at most $n$ rounds.

We also introduce the family of {\em proper} relaxations 
which generalizes to its logical extreme  the  classic
star relaxation, an equivalent form of the natural LP. 
We characterize the behavior  of proper relaxations for both
\lbfl\  and \cfl\ through  a sharp
threshold phenomenon under which the integrality gap drops from unbounded to 1.
\end{abstract}

\clearpage
\setcounter{page}{1}

\section{Introduction}
\label{intro}

Facility location problems have been studied extensively in 
 operations research,
mathematical programming, and theoretical computer science. 
The 
{\em uncapacitated facility location (\ufl)} problem is defined as follows. 
A set $F$ of 
facilities  and a set $C$ of  clients  are given. 
Every client has to be assigned to an opened facility. 
Opening a facility
$i$ incurs a nonnegative cost $f_i,$ while assigning a client $j$ to facility
$i$ incurs a nonnegative connection cost  $c_{ij}.$  The goal is to open a subset 
$F' \subseteq F$ of the facilities and  assign each
client to an open facility so that  the total cost is minimized. 
Hochbaum gave a greedy  $O(\log |C|)$-approximation algorithm
\cite{Hochbaum82}. By a straightforward reduction from Set Cover this
cannot be improved, unless {\sf P = NP} \cite{RazS97}. 

In the {\em metric} \ufl\ 
the connection  costs satisfy the following variant of the triangle inequality:
$c_{ij} \leq c_{ij'} + c_{i'j'} + c_{i'j}$ for any $i, i'\in F$ and $j, j' \in C.$
The first constant-factor approximation of $3.16$ was given by Shmoys, 
Tardos and Ardaal \cite{ShmoysTA97}. Over the years,
\ufl\ has served as a prime testbed for several techniques in
the design of approximation algorithms (see, e.g., 
\cite{ShmoysWbook}). Among those techniques LP-based
methods, such as filtering, randomized rounding and the primal-dual method,
 have been particularly prominent and have yielded several improved bounds.
After a long series of papers 
the currently best approximation ratio
for metric \ufl\ is $1.488$ \cite{Li11}. 
Guha and Khuller \cite{GuhaK99} proved that there is 
no $\rho$-approximation algorithm for metric \ufl\ with 
$\rho < 1.463$ unless ${\sf NP} \subseteq {\sf
DTIME}(n^{O(\log \log n)})$ using Feige's hardness result for Set
Cover \cite{Feige98}.
Sviridenko (see \cite{Vygen05}) showed that the lower bound holds
unless {\sf P = NP.}
In this paper we focus on two generalizations of the metric \ufl:  the 
 {\em capacitated facility location (\cfl\/)} and
 the {\em lower-bounded facility location (\lbfl\/)} problems.
To our knowledge  the $1.463$ lower bound  is the only inapproximability result known
for these two as well.

\cfl\/ is the  generalization of metric \ufl\ where every facility $i$ has a
capacity $u_i$  that specifies  the maximum number of clients that may
be
assigned to $i.$ In {\em uniform} \cfl\ all facilities have
the same capacity $U.$  
  Finding an approximation algorithm for  \cfl\/ that uses a linear programming
lower bound, or even proving a constant integrality gap for an
efficient  LP
relaxation, are notorious open problems. Intriguingly, 
the  natural LP relaxations  have 
an unbounded integrality gap and
the only known  $O(1)$-approximation algorithms are
based on local search.  The currently best ratios 
for the non-uniform  and the uniform case are $5$ \cite{BansalGG12} and $3$  
\cite{AggarwalLBGGJ12} respectively. 
Compared to local search, relaxations have 
the distinct advantage that they provide, on an instance-by-instance basis,
a  concrete  lower
bound on the optimum. A small  gap between the
optimal integer and fractional solutions  could be exploited
to speed up an exact  computation. From the viewpoint of
approximation, 
comparing the LP optimum against the  solution  output by an
LP-based algorithm establishes 
a  guarantee than is at least as strong as the one established  a priori by  worst-case
analysis.
In contrast, when a  local search algorithm terminates, 
it is not at all clear what the lower
bound is. 
According to Shmoys and Williamson \cite{ShmoysWbook} devising 
a relaxation-based algorithm for \cfl\  is one of
the top $10$ open problems in approximation algorithms.

 \lbfl\ is in a sense the opposite problem to \cfl\ and 
was introduced independently by Karger and Minkoff \cite{KargerM00} and Guha et
al. \cite{GuhaMM00} in the context of network design problems with buy-at-bulk
features.  In an instance of \lbfl\ every facility $i$ comes with a  lower
bound $b_i$ 
which is the minimum number of clients that must be assigned
 to $i$  if we open it. In {\em uniform} \lbfl\ all the lower bounds
have the same value $B.$  \lbfl\ is even less well-understood than \cfl. 
 The first true approximation algorithm for the uniform case 
was given in \cite{Svitkina08} with a performance guarantee of
$448,$ which has been recently improved  to $82.6 $ \cite{AhmadianS12}. 
Interestingly, the \lbfl\ algorithms
from \cite{Svitkina08,AhmadianS12}  both use a \cfl\ algorithm on a
suitable instance as a subroutine.

Studying the limits of linear programming  relaxations for 
intractable problems is
an active area of research. 
The inherent challenge in this work is  to characterize collections
of LPs 
for which no explicit description is known. 
The
main  direction is to lower bound 
the size of extended formulations
that express optimal or near-optimal solutions,
or to determine
the  integrality gap of  comprehensive families  of
valid LP relaxations. 
  Yannakakis \cite{Yannakakis91}  proved early on that 
any 
symmetric linear relaxation that expresses the Traveling Salesman
polytope must have exponential size. 
 Recent results lift the symmetry assumption 
\cite{FioriniMPTD12}   and  characterize  the size of LPs that
express approximate solutions  to Clique \cite{BraunFPS12,BravermanM13}.

A lot of effort has been  devoted to 
understanding the quality of relaxations obtained by an iterative 
lift-and-project procedure. Such procedures define hierarchies of 
successively  stronger relaxations, where  valid inequalities are added at 
each level. After at most $n$ rounds, where $n$ is the number of 
variables, all 
valid inequalities have been  added and thus the integer polytope is
expressed. 
Relevant methods  include those  developed  by Balas et
al. \cite{BalasCC93},  Lov\'{a}sz and Schrijver \cite{LovaszS91} (for
linear and semidefinite programs, denoted respectively LS and LS$_{+}$), 
Sherali and Adams \cite{SheraliA90} (denoted SA),    Lassere  \cite{Lasserre01}
(for semidefinite programs). 
See
\cite{laurent} for a comparative discussion.  
Exploring the
structure of the successive relaxations  in a hierarchy 
is of intrinsic interest in polyhedral
combinatorics. The seminal work of Arora et al.
\cite{AroraBL02,AroraBLT06} introduced the use of 
hierarchies as a model of computation for obtaining  hardness of
approximation results. 
Proving that the integrality 
gap for a problem remains large  after many 
rounds   is an unconditional  guarantee against the  class of
sophisticated relaxations
obtained through the specific procedure. 
Despite the amount of effort, the effect on approximation  of the
 different hierarchies is not well-understood. Vertex Cover is a
  prominent case among the problems studied early on. 
Arora et al. \cite{AroraBLT06} showed that after
$\Omega(\log n)$ rounds of the LS procedure the
integrality gap for Vertex Cover remains $2-\epsilon.$ Schoenebeck at
al. \cite{SchoenebeckTT07} proved that the $2-\epsilon$ gap survives
for $\Omega(n)$ rounds of  LS. The body of 
work on   hierarchies keeps growing, see, e.g., 
\cite{GeorgiouMPT10,FernandezdlVKM07,SchoenebeckTT07b,CharikarMM09,MathieuS09,Tulsiani09,BhaskaraCVGZ12}. Some
of
those results examine  semidefinite relaxations, a direction we do not
 pursue here. 

Investigating the strength  of linear relaxations  is driven by the 
perception of LP-based algorithms as a  powerful paradigm for designing
approximation algorithms. 
Recent work inspired from 
\cite{Raghavendra08} 
explores a complementary direction: how to translate integrality gaps for LPs into {\sf UGC}-based
hardness of approximation results \cite{KumarMTV11}.

In  recent work,  improved  approximations  were given  for
$k$-median   \cite{LiS13}      and   capacitated   $k$-center
\cite{CyganHK12,AnBS13},  problems closely  related to
facility location.  For both, the improvements are obtained 
by LP-based techniques
that  include preprocessing  of the  instance in  order to  defeat the
known integrality gap. For $k$-median, the authors of \cite{LiS13}
state that their 
$(1+\sqrt{3} + \epsilon)$-approximation algorithm   can
be converted  to a  rounding algorithm on an 
 $O(\frac{1}{\epsilon^2})$-level LP  in the SA hierarchy.  In \cite{AnBS13}
the authors raise as an important question  to understand  the
power  of lift-and-project methods  for capacitated  location
problems, including  whether they 
automatically capture such preprocessing steps.

\subsection{Our results}
In this paper we give the first characterization  of 
 the    integrality gap  for families of 
 linear relaxations  for metric \cfl\ and \lbfl\  and thus provide the
 first results on the open problem of \cite{ShmoysWbook}. 
We study two substantial families of
strengthened LPs. 
Our derivations make  no time-complexity assumptions 
and are thus unconditional.  We also partially  answer 
the question of \cite{AnBS13} for \cfl:  if there is an efficient relaxation, 
it is not captured even after a linear number of rounds in
the LS hierarchy.

We first introduce    
the family of  proper relaxations
which are ``configuration''-like linear programs.
The so-called Configuration LP was  used by 
Bansal and Sviridenko 
\cite{BansalS06} for the Santa Claus problem and has yielded valuable insights
and improved results, mostly
for resource allocation  and scheduling problems
(e.g., \cite{Svensson11, AsadpourFS11,
HaeuplerSS11, SviridenkoW13}).
A configuration in a scheduling setting usually refers to a set of
jobs $J_i$ that
can be feasibly assigned to a given  machine $i$ while meeting some load
constraint. A typical Configuration LP has therefore    an exponential number of
variables.    
The analogue of the Configuration
LP for facility location already exists (see, e.g., \cite{JainMMSV03}): it is the well-known 
{\em star relaxation,} in which every  variable corresponds to a {\em star,} i.e., a
facility $f$ and a set
of clients assigned to $f.$ 
The natural star relaxation 
for   \cfl\ and \lbfl\  is  equivalent to the standard LPs 
so it has an unbounded integrality gap.  
We take the idea of a star  to its logical  extreme by 
introducing  classes. 
A {\em class} consists of a set with an arbitrary number of facilities and clients
together with an 
assignment of each client to a facility in the set. 
The definition
of a class can thus vary from simple, ``local'' 
assignments of  some clients to  a  single facility, to  
 ``global'' snapshots  of the instance that 
express  the assignment  of many
clients to a large set of  facilities.  
A {\em proper relaxation} for an instance is defined by a collection
$\mathcal{C}$ of classes and a decision variable for every class. 
We allow great freedom in 
defining  ${\cal C}\colon$ 
the only requirement   is that the resulting
formulation is symmetric and valid. 
The {\em complexity $\alpha$} of a proper relaxation is the maximum fraction
 of the 
available facilities that is contained in a class of $\mathcal{C}.$
Proper LPs are stronger than the standard relaxation. 
One can
construct infinite 
families  of instances 
where, by  increasing the complexity in a proper relaxation, one cuts off  more
and more fractional solutions.  
In this sense, all proper LPs for an instance can
be  thought of as forming  a (non-strict)    hierarchy, with the star
relaxation at the lowest level. 
We characterize their behavior  through a threshold result: 
anything less than maximum complexity results in unboundedness of
the integrality gap, while there are
proper relaxations of maximum complexity with an integrality gap of
$1$.  In the latter, $\mathcal{C}$ corresponds simply to 
the set of all integer feasible solutions. 
Our precise results 
are the following theorems. Their proofs rely 
on the symmetry of the formulations. 

\begin{theorem}
\label{theorem:proper}
Every proper relaxation for uniform \lbfl\ with complexity $\alpha < 1$ has an
unbounded integrality gap of $\Omega(n)$ where $n$ is the number of
facilities. There exist proper relaxations of complexity $\alpha=1$ that have an
integrality gap of $1.$ 
\end{theorem}

\begin{theorem}
\label{theorem:proper2}
Every proper relaxation for uniform \cfl\ with complexity $\alpha < 1$ has an
unbounded integrality gap of $\Omega(n^2)$ where $n$ is the number of
facilities. There exist proper relaxations of complexity $\alpha=1$ that have an
integrality gap of $1.$ 
\end{theorem}

The second family we investigate consists of  linear  relaxations resulting from
repeated applications of the 
LS procedure starting from the natural LP relaxation for  \cfl\/.
We show that a specific bad solution with unbounded integrality gap 
 survives $\Omega(n)$ rounds of LS.
The solution is defined on an instance $I$ 
 with $n$ facilities and $m=\Theta(n^4)$ 
clients.

It is  well-known that the LS procedure 
 extends to mixed 0-1  programs \cite{LovaszS91,BalasCC93} such as  \cfl
 \ with general client demands.  In that case the convex hull of the
 mixed-integer feasible set is known to be obtained the latest at the $p$th
 level of the LS hierarchy, where $p$ is the number of binary variables
 (\cite{LovaszS91},  \cite[Theorem~2.6]{BalasCC93}). For  \cfl,   
$p$ equals the number $n$ of facilities. 
In our instance 
$I,$ the clients have unit demands and as such the integer and
mixed-integer versions of the problem are 
equivalent. 
 In the lifting procedure, we treat  both the facility opening and the assignment
 variables as binary. 
It is easy to see that in  every round we obtain a polytope which is at least as tight as the one obtained 
when only the facility-opening variables are binary. 
 Therefore  our lower bound of $\Omega(n)$  applies also 
 to the mixed-integer LS lifting procedure and is linear in the parameter $p.$ 
Our proof is via protection matrices \cite{LovaszS91}. 
Using a simple reformulation of LS we give an explicit, fully
constructive definition of the matrices generated at each level that
witness the survival of the bad  fractional solution. 
The result is the following. 

\begin{theorem}
\label{theorem:ls-cfl}
For every sufficiently large $n,$  
there is an instance of uniform \cfl \   with $n$ facilities and $\Theta(n^4)$ clients 
so that the integrality gap 
 after $\Omega(n)$ rounds of the LS procedure is  $\Omega(n).$  
\end{theorem}


\subsection{Other related work}
Koropulu et al. \cite{KoropuluPR00} gave the first constant-factor approximation
algorithm for uniform \cfl. Chudak and Williamson
\cite{ChudakW05} obtained a ratio of  $6,$ subsequently  improved to $5.83$
\cite{CharikarG99}.
 P\'{a}l et al. \cite{PalTW01} gave
the first constant-factor approximation for non-uniform \cfl. This was improved
by Mahdian and P\'{a}l \cite{MahdianP03} and Zhang et al. \cite{ZhangCY04} to a
$5.83$-approximation algorithm. As mentioned, the currently best guarantee is 5,
due  to Bansal et al. \cite{BansalGG12}. All these approaches use local search. 

Levi et al.  \cite{LeviSS12} gave a  5-approximation 
algorithm, based on the standard LP, for the special case of \cfl\ where
all facilities have the same opening cost. 
In the {\em soft-capacitated} facility location problem one is allowed to open
multiple copies of the same facility. 
Work on this problem 
includes \cite{ShmoysTA97,ChudakS99, ChudakW05, JainV01}.  
As observed in \cite{JainMMSV03}  a $\rho$-approximation for  \ufl\ yields a
$2\rho$-approximation for the case with soft capacities. Mahdian, Ye and Zhang
\cite{MahdianYZ03} noticed a sharper tradeoff and obtained a
$2$-approximation. A tradeoff between the blowup  of capacities and
the cost approximation for \cfl\ was studied in
\cite{AbramsMMP02}. Bicriteria approximations for \lbfl\  appeared in
\cite{KargerM00,GuhaMM00}. 

For hard capacities and general demands the feasiblity of the
{\em unsplittable} case, where the demand of each client has to be assigned to a
single facility, is NP-complete, as {\sc Partition} reduces to it. 
Bateni and Hajiaghayi \cite{BateniH12} considered the unsplittable problem 
with an $(1+\epsilon)$ violation of the capacities and obtained an
$O(\log n)$-approximation.

\bigskip

The outline of this paper is as follows. In Section~\ref{sec:prel} we give
preliminary definitions and in Section~\ref{sec:firstfamily} we introduce the proper
relaxations. The proofs of Theorems \ref{theorem:proper}, \ref{theorem:proper2} are  in Sections
\ref{sec:proof_theorem_p1} and \ref{sec:proof_theorem_p2} respectively. In Section~\ref{sec:ls} we present the
necessary background for the Lov\'{a}sz-Schrijver procedure. The proof of
Theorem~\ref{theorem:ls-cfl} is in Section~\ref{sec:ls_cfl}. We conclude with a 
discussion of our results in Section~\ref{sec:open}.

\section{Preliminaries}
\label{sec:prel}

Given an instance $I(F,C)$ of \cfl\ or \lbfl, we use $n, m$ to denote $|F|$ and $|C|$
respectively.
We will show our negative results for uniform, integer, capacities and lower
bounds. Each client can be thought of as representing one unit of demand.
 It  is  well-known  that in such a setting  the  splittable  and
unsplittable versions  of the problem are equivalent. 
The following 0-1  IP is the standard  valid formulation of uncapacitated 
facility location with unsplittable unit demands.

\begin{align}
\min \sum_{i \in F} f_iy_i + \sum_{i \in F}\sum_{j
  \in C} x_{ij}c_{ij}   & &    \label{obj} \\   
x_{ij} \leq y_i   &&  \forall i \in F, \forall j \in C   \label{x<y}  \\
\sum_{i \in F} x_{ij} =1   &&   \forall j \in C   \label{eq}  \\
x_{ij} \in \{0,1\} &&  \forall i \in F, \forall j \in C  \label{eq:intx} \\
y_i \in \{0,1\}   && \forall i \in F   \label{eq:inty}
\end{align}

We give two well-known relaxations  for the problem.  The first one is the
natural  LP resulting from the above IP by replacing the integrality constraints
 with: 

\begin{align}
0 \leq x_{ij} \leq 1 &&  \forall i \in F, \forall j \in C  \label{1>x>0}  \\
0 \leq y_i \leq 1  && \forall i \in F    \label{eq:nni}
\end{align} 

To obtain the standard LP relaxations for 
uniform \cfl\ and \lbfl\ the following constraints are added
respectively:

\begin{align}
\sum_ {j} x_{ij} \leq U y_i    && \forall i \in F   \label{sat1}\\
\sum_ {j} x_{ij} \geq By_i    && \forall i \in F   \label{sat2}
\end{align} 

In the rest of the paper we slightly abuse terminology by 
using  the term {\em (LP-classic)} for both LPs. It will be clear from the context to
which problem we refer (\cfl\  or \lbfl). 

The second well-known LP  is the star relaxation.
A {\em star} is a set consisting of some  clients and
one facility. Let $\mathcal{S}$ be a set of stars. For a star  $s\in \mathcal{S},$ let  $x_s$ be an indicator variable denoting
whether  $s$ is picked.   The cost  $c_s$ of star $s$   is equal
to the opening cost of the corresponding facility plus the cost of
connecting the star's clients to it. 

\begin{align}
\min \sum_{s} c_sx_{s}   &&   \tag{LP-star} \\
\sum_{s \ni  j } x_s = 1   && \forall j \in C  \label{star:client} \\
\sum_{s \ni  i } x_s \leq  1   && \forall i \in F \label{star:facility} \\ 
x_s \geq 0       &&  \text{for all stars } s \in \mathcal{S} \label{star:nn} 
\end{align} 

Defining $\mathcal{S}$ as the set of all stars $s$ where  the total number of the clients in $s$ 
is at most the
capacity $U$ (at least the bound $B$),  we get corresponding relaxations for
  \cfl\ (\lbfl).
In the rest of the paper we  slightly abuse terminology by 
using    {\em (LP-star)}  to refer to the
star relaxation for the problem we examine each time (\cfl\  or \lbfl).

It is well known  that for both \cfl\ and \lbfl, (LP-classic) and (LP-star) are
equivalent, therefore (LP-star) can be solved in  polynomial time.
 For
the sake of completion we include the relevant 
Lemma~\ref{lemma:ap-classic} in the Appendix.

\section{Proper Relaxations}     \label{sec:firstfamily}

In this section  we introduce  the family of  proper relaxations.

Consider a   $0$-$1$ $(y,x)$ vector on the set of
 variables  of the  classic  relaxation (LP-classic)
such that $y_i \geq x_{ij}$ for all $i \in F, j \in C.$  The meaning  of
 $y_i=1$ is the usual one that
 we open facility $i.$  Likewise, the meaning of $x_{ij}=1$ is
 that we assign client $j$ to facility $i$. We call such a vector a 
 \emph{class}. Note that  the definition is quite general  and a class
 can be defined  from any such $(y,x)$, which may or  may not have a
 relationship  to a  feasible  integer solution.  
Classes generalize the notion of a star. 
We  denote the  vector
 corresponding  to a  class $cl$  as $(y,x)_{cl}$.  We  associate with
 class  $cl$   the  {\em cost   of  the  class}   $c_{cl}=\sum_{i \mid
   y_i=1  \in
   (y,x)_{cl}} f_i+  \sum_{i,j \mid  x_{ij}=1 \in  (y,x)_{cl}}
 c_{ij}$. Let the {\em assignments of class} $cl$ be defined as 
 $Agn_{cl}= \{(i,j) \in F\times C \mid  x_{ij}=1$ in $(y,x)_{cl}\}$.
We say that $cl$ {\em contains}  facility $i,$ if the corresponding entry
$y_i$  in the vector $(y,x)_{cl}$ equals $1.$  
The set of facilities contained in $cl$ is denoted by 
 $F(cl).$

\begin{definition}  {\bf (Constellation LPs)} \label{def:constell}
Let $\mathcal{C}$ be a set of classes defined for an instance $I(F,C)$
of
\lbfl. Let $x_{cl}$  be a variable associated with class $cl \in
\mathcal{C}.$
The {\em  constellation LP with class set}
  $\mathcal{C},$ denoted LP($\mathcal{C}$),  is defined as

\begin{align*} 
\min \sum_{cl \in \mathcal{C}} c_{cl}x_{cl} &&  \tag{LP($\mathcal{C}$)} \\
\sum\nolimits_{{cl} \mid  \exists i:(i,j) \in  Agn_{cl}}
x_{cl}=1 && \forall j \in C   \\
\sum\nolimits_{{cl} \mid  i  \in  F({cl})} x_{cl} \leq 1 &&
\forall i \in F  \\
 x_{cl} \geq 0  && \forall cl \in \mathcal{C} 
\end{align*} 

\end{definition}
In what follows we will refer simply to a constellation LP when
$\mathcal{C}$ is implied from the context.  
We define the \emph{projection} $s'=(y^{s'}, x^{s'})$ of solution $s=(x^s_{cl})_{cl \in \mathcal{C}}$ 
of a  constellation LP to the classic facility opening
and assignment variables $(y,x)$ as $y_i^{s'}=\sum_{cl|i\in cl}x_{cl}^{s}$ and 
$x_{ij}^{s'}=\sum_{cl| (i,j) \in Agn_{cl}}x_{cl}^{s}$.

We will restrict our attention to  constellation LPs that satisfy a  natural property:   the LP is symmetric
 with respect to the clients and  the  facilities. 
The fact  that all facilities have the  same capacity / lower bound and all
clients have unit demand makes  this  property quite sound. 
For a class   $cl$ and
$f_1: \{1,...,n \} \rightarrow \{1,...,n \}$ a permutation of the facilities, we denote by $cl_{f_1}$
the class resulting by exchanging  for all $k, j$ the values  of the $y_{k}$  and
$x_{kj}$ coordinates 
of $(y,x)_{cl}$  with   the value of  the $y_{f_1(k)}$  and $x_{f_1(k)j}$ coordinates of
$(y,x)_{cl}$. Similarly,  for $f_2: \{1,...,  m\} \rightarrow \{1,...,
m\}$ a permutation of the clients, we denote  by $cl_{f_2}$ the class resulting  by exchanging for
every $i$  the value of  the $x_{ik}$ coordinate of  $(y,x)_{cl}$ with
 the value of the $x_{if_2(k)}$ coordinate of $(y,x)_{cl}$.

\begin{definition} {\bf ($P_1$: Symmetry)} \label{def:symmetry}
We  say  that property  $P_1$  holds for the constellation linear program LP($\mathcal{C}$)   if  the
following is  true: let $\phi  (n)$ be any  permutation of $F$  and $\mu
(m)$ any permutation of $C$.
 Then, for every  class $cl \in \mathcal{C},$  $cl_{\phi}$ 
and $cl_{\mu}$ are also in $\mathcal{C}.$
\end{definition}

The second property we require is the obvious one that the 
relaxation  is {\em valid,} i.e., the projection of its  feasible region to 
$(y,x)$ contains all the characteristic vectors of the feasible integer solutions of the instance.

\begin{definition}  {\bf (Proper Relaxations)}  \label{def:proper} 
We call {\em proper relaxation}  for \cfl\ (\lbfl\/)  a constellation LP
 that is valid and satisfies property $P_1.$ 
\end{definition}

Relaxation (LP-star) is
obviously a proper relaxation, while (LP-classic) is equivalent to
(LP-star). Therefore proper relaxations generalize the known natural
relaxations for \cfl\ and \lbfl.

\subsection{Complexity of proper relaxations}
\label{proper proof}

Our main result on proper relaxations is that     proper LPs that  are not
``complex'' enough have an unbounded integrality gap while those  that
are sufficiently ``complex'' have an integrality gap of $1.$  To that end, we
define  the complexity  of  a  proper LP 

\begin{definition}   \label{def:complexity}
Given an instance $I(F,C)$ of \cfl\ (\lbfl\/)
let $F'$ be  a 
maximum-cardinality set  of open facilities in an integral feasible
solution. The {\em complexity $\alpha$} of a  proper
relaxation $LP(\mathcal{C})$ for $I$ is
defined as the $\sup_{cl \in \mathcal{C}} \frac{|F(cl)|}{|F'|}.$ 
\end{definition}

Note that for \lbfl\ it is possible  to have a proper relaxation with complexity greater
than $1$. 
The  complexity of  a  proper LP  represents the  maximum
fraction of the  total number of feasibly openable  facilities that is
allowed in a single class. 
For   a  proper relaxation  $LP(\mathcal{C}),$ the  complexity
describes to what extent 
 classes in $\mathcal{C}$ consider the instance locally. 
A complexity of nearly $1$
means that there are classes that take into consideration almost the whole instance
at once, while a low complexity means that all classes consider
the assignments of a small fraction of the instance at a time.  
By increasing the complexity of a proper LP  for a given instance 
 we can produce strictly stronger 
proper relaxations. A simple example is given below.

\begin{example}\label{proper_str}
An increased complexity allows strictly stronger proper relaxations.
\end{example}

First we show how one can construct any integer solution using classes that open the
same number of facilities.
Consider an integer solution $s$ with opened facilities $1,...,t$. We will use the following classes 
in which exactly $r<t$ facilities are opened:
For any set of  $t$ consecutive classes in a cyclic ordering, namely $(1,...,r),(2,...,r+1),...,(t,...,r-1)$, define a class that opens those facilities and makes the same assignments to them 
as $s$. Then the integer solution is obtained  if for every $cl$ we set $x_{cl}=1/r$.
Observe that the latter solution is feasible for the proper relaxation.

We give a toy example showing that by increasing the complexity, we can
get strictly stronger relaxations. Consider an \lbfl\ instance with $4$ facilities $2$ sets $S_1,S_2$
of $13$ clients each and 2 sets $S_3,S_4$ of $9$ clients each and $B=10$.  For the star relaxation
(complexity $\alpha=1/4$ for this instance)
there is a feasible solution $\bar{s}$ whose projection to $(y,x)$
 is the following $(\bar{y},\bar{x})$: for facility $1,$ $\bar{y}_1=1$ and is assigned $S_1$ integrally, for facility $2,$ $\bar{y}_2=1$ and is assigned $S_2$ integrally, for facility $3,$ $\bar{y}_3=9/10$ and is assigned each client of $S_3$ with a fraction of $9/10$ and each of $S_4$ with $1/10$, and similarly for facility $4,$ $\bar{y}_4=9/10$ and is assigned
each client of $S_4$ with a fraction of $9/10$ and each of $S_3$ with $1/10$. Actually
a direct consequence of Theorem \ref{theorem:proper} is that for any proper relaxation of the same complexity as the star relaxation, the above solution is feasible.

Now consider the following proper relaxation: all characteristic vectors 
of integer solutions with at most
$3$ facilities are classes plus all the 
vectors of solutions with $4$ facilities restricted in any $3$ facilities ($3/4$ parts of integer solutions that open all four facilities).
It is symmetric and valid by the previous discussion and has complexity $\alpha=3/4$. 
In any assignment of values to the class variables  that projects to $(\bar{y},\bar{x})$ the following are true:
since classes with less than $3$ facilities are integer solutions, they contain
assignments for all the clients and thus if we were 
to use a non-zero measure of such classes we would make non-zero assignment 
that does not exist in the support of $(\bar{y},\bar{x})$.
 If we use
classes with exactly $3$ facilities, then exactly one of facilities $3,4$ must be present, 
since no integer solution opens them both with just the clients in $S_3 \cup S_4$. 
So we have to use at least $\bar{y_3}+\bar{y_4}=18/10$ measure of such classes. 
So each one of facilities $1,2$
must be present in more than a unit of classes, which would make the solution infeasible.

\medskip

If we allow the complexity to be $1$, then one can find
proper relaxations that have integrality gap equal to $1.$ 


\begin{theorem}  \label{thm:gap1}
There is a proper LP relaxation for \cfl\ (\lbfl\/)  that has complexity $1$ 
and whose projection to $(y,x)$ expresses the
integral polytope. 
\end{theorem}

{\noindent{\em Proof of Theorem \ref{thm:gap1}.}}
For a given instance let $\mathcal{C}$ consist of a class for each
distinct integral solution. The resulting $LP(\mathcal{C})$ is clearly
proper. Let $x$ be any feasible solution of $LP(\mathcal{C})$ and let
$S$ be the support of  the solution. For every $cl \in S,$ and 
for every client $j \in C,$ there is an $i \in F,$ such that $(i,j)
\in Agn_{cl}.$ Therefore 
$$ \sum_{cl \in S} x_{cl} = 1.$$
This implies that $x$ is a convex combination of integral
solutions. By the boundedness of the feasible region of
$LP(\mathcal{C}),$ the  corresponding polytope is integral.  
\mbox{} \hfill \mathqed  

Clearly not every LP with complexity $1$ has an integrality gap of $1$
since it might contain weak classes together with  strong
ones.

We proceed  to show that a complexity of $1$  is necessary 
in order to avoid a dramatic drop in solution quality as stated in 
Theorems \ref{theorem:proper} and \ref{theorem:proper2}.

\section{Proof of Theorem \ref{theorem:proper}}
\label{sec:proof_theorem_p1}

Our proof includes the following steps. We define an instance $I$ 
and consider any proper relaxation $LP(\mathcal{C})$ for $I$ that has complexity
$\alpha <1.$ 
Given $\alpha,$ we use   the validity  and symmetry properties to show the existence of
a specific set of classes in $\mathcal{C}$. Then we use these classes to construct a
desired feasible fractional solution, relying again on symmetry. 
In the last step  we specify  the distances between the clients and  the facilities, so
that the instance is metric and the constructed solution proves unbounded integrality
gap.

\subsection{Existence of a certain type of classes}

Let us fix for the remainder of the section 
an instance $I$ with $n+1$ facilities, where $n$ is
sufficiently large to ensure  that $\alpha n \leq n - c_0$  where
$c_0,$  is a
constant greater than or equal to  $2$. Let the bound $B=n^2$, and let
the number of  clients be $n^3$. Notice that  there are enough clients
to open $n$ facilities, with  exactly $n^2$ clients assigned  to each
one that is opened. The  facility costs  and the assignment  costs will  be defined
later.  Recall  that the  space  of  feasible  solutions of  a  proper
relaxation is independent of the costs.

 We  assume that  the  facilities are  numbered
$i=1,2,...,n+1$. 
For a solution $p$ we  denote by $Clients_p(i)$
the set  of clients  that are assigned  to facility $i$  in solution
$p$, and  likewise for a  class $cl$ we denote  by $Clients_{cl}(i)$
the set of clients that are assigned to facility $i$.  
Consider  an  integral  solution  $s$  to  the  instance  where 
facilities $1,...,n$  are opened. 
Since our proper  relaxation is valid, it must have   a feasible 
solution $s'=(x_{cl})_{cl \in \mathcal{C}}$  whose projection to $(y,x)$ gives the characteristic
vector of $s$.  We prove the existence of a  class $cl_0,$ with some desirable
properties, in the support of $s'.$ 

By Definition~\ref{def:constell},
$s'$ can only be obtained as a positive combination of classes $cl$ such that for
every    facility   $i$    we   have    $Clients_{cl}(i)   \subseteq
Clients_s(i)$, Otherwise,  if the variables  of a  class $cl$
with $Clients_{cl}(i)  \setminus  Clients_s(i) \neq \varnothing$ have  non-zero value,
then in $s'$ there will be  some client assigned to
some facility with a positive fraction, while the projection of $s',$ namely $s,$
does  not include  the
particular  assignment.  
Moreover,  since exactly $B$ clients are  assigned to each
facility in  $s$,   for every facility $i$
that   is  contained   in  such   a  class   $cl,$  $Clients_{cl}(i)=
Clients_s(i)$. To see why this  is true, 
since in  $s$ we have $y_i=1,$ for all $i\leq n,$   it follows  that for every facility $i\leq n$,
 $\sum_{cl  | \exists (i,j)  \in  Agn_{cl}} x_{cl}  =1$.  
But  then  we have  that
$|Clients_s(i)|=B=\sum_{cl  |   \exists  (i,j)  \in  Agn_{cl}}
x_{cl}|Clients_{cl}(i)|$.  We have already established   that $x_{cl}>0
\implies |Clients_{cl}(i)|  \leq B$. Then $B$ is  a convex combination
of quantities less than or equal to $B$, so for all such classes $cl$
we have $|Clients_{cl}(i)|=B$.

Therefore
in the class set of any proper relaxation for $I,$ there is 
a class $cl_0$ that assigns exactly $B$ clients to each of  the
facilities in $F(cl_0).$ By the value of $\alpha,$ 
 $|F(cl_0)| \leq n -c_0.$ The following
lemma has been proved.

\begin{lemma}   \label{lemma:existence} 
Given the specific instance $I,$ any proper relaxation  of complexity
$\alpha$ for $I$ contains in its class set a class 
$cl_0$ that assigns $B$ clients to each of $n-c$ facilities, for some
integer  $c \geq 2.$ 
\end{lemma}

\subsection{Construction of a bad  solution}  
\label{subsec:badlbfl}

In the present section we will use the class $cl_0$ along with the
symmetric classes to construct a solution to the proper LP with 
the following
property: there are some  $q$ facilities   that
are almost integrally opened while the number of distinct  clients assigned to them will be less than $Bq$. 

Recall that by property $P_1$ every class that is isomorphic to $cl_0$ is
also a class of our proper relaxation. This means that
every set  of $n-c$ facilities and every  set of $B(n-c)$ clients
assigned to those facilities so that each facility is assigned exactly
$B$ clients, defines a class, called {\em admissible,} that belongs to the set of classes
defined of a  proper relaxation for the instance $I$.

Let  us  turn  again  to  the solution  $s$  to  provide  some  more
definitions.  For   every  facility  $i,$   $i=1,...n-1$,  we  choose
arbitrarily a client $j'$ assigned to  it by $s$. For each such facility
$i$  we   denote  by  $Exclusive(i)$   the  set  of   clients  $
Clients_s(i) - \{j' \},$ i.e., the set of clients assigned to
$i$ by $s$ after we discard $j'$ (we will also call them the
{\em exclusive clients of $i$}). For facilities $n,$ $n+1$ the sets
$Exclusive(n),$ $Exclusive(n+1)$ are identical  and defined to be equal to 
 the union of $Clients_s(n)$ with all
the  discarded clients from  the other  facilities. In  the fractional
solution that we will construct below, the clients in $Exclusive(i)$
will be almost integrally assigned to $i$ for $i=1,...,n-1$.

We  are   ready  to  describe  the  construction   of  the  fractional
solution. We will use a subset $S$ of admissible classes that 
do not contain both $n$ and  $n+1$. $S$ contains all such classes   
 $cl$ that assign to each facility $i \leq
n-1$  in the class  the set  of clients  $Exclusive(i)$ plus  one more
client selected  from the sets $Exclusive(i')$  for those facilities
$i' \leq n-1$ that do not belong  to $cl$ (there are at least $c-1$
of them). As for facility $n$ (resp. $n+1$), if it  is contained in $cl,$ then
it is  assigned some set  of $B$ clients  out of the total  $B+n-1$ in
$Exclusive(n)$ (resp. $Exclusive(n+1)$).  
All classes not in $S$  will get a value
of zero in our solution.
We
will distinguish the classes in $S$ into two types: the classes
of {\em type $A$} that contain facility  $n$ or $n+1$ but not both, and classes
of {\em type $B$} that
 contain neither $n$ nor $n+1$.

We consider  first classes of type  $A$. We give  to each such class   a
very small  quantity of  measure $\epsilon$. Let  $\phi$ be  the total
amount of measure used. We call this step $Round_A$.  The
following  lemma shows  that after  $Round_A$, the  partial fractional
solution  induced  by  the  classes  has a  convenient  and  symmetric
structure:

\begin{lemma}  \label{lemma:roundA}
After  $Round_A$,  each client  $j  \in  Exclusive(i),$  $i\leq n-1,$  is
assigned to $i$  with a fraction of $\frac{n-c-1}{n-1}  \phi$ and is
assigned to each other facility $i',$ $i' \neq i,$ $i' \leq  n-1,$ with a
fraction of $\frac{n-c-1}{(n-1)(n-2)(n^2-1)} \phi$. Each client $j \in
Exclusive(n)$ ($=Exclusive(n+1)$)   is   assigned   to   $n$ and to $n+1$ 
 with   a   fraction   of $\frac{n^2}{2(n^2+n-1)} \phi$.
\end{lemma}

\begin{proof}
Consider a facility  $i, i\leq n-1$. Since exactly one of facilities  $n,n+1$ is present in
all the classes of type  $A$ and each class contains $n-c$ facilities,
$i$ is  present in the  classes of $Round_A$  $\frac{n-c-1}{n-1}$ of
the time due to symmetry of the classes. Each time $i$ is present in
a class  $cl$ that  class $cl$ assigns  all $j \in  Exclusive(i)$ to
$i$.  So  client  $j$  is  assigned  to $i$  with  a  fraction  of
$\frac{n-c-1}{n-1}  \phi$. When $i$  is not  present in  class $cl$,
which happens  $\frac{c}{n-1}$ of the time, then  its exclusive clients
along with  the exclusive  clients of all  the other  $c-1$ facilities
that  are also  not  present in  $cl$  are used  to  help the  $n-c-1$
facilities $i \leq n-1,$  reach the bound $B$ of clients (recall
that the number of exclusive clients of each such facility is equal to
$B-1$).  Each time  this happens, the $n-c-1$ facilities  in $cl$ need
$n-c-1$  additional clients, while  the exclusive  clients of  the $c$
facilities that are  not present in $cl$ are  $c(n^2-1)$ in total. Due
to symmetry  once again, a  specific client $j \in  Exclusive(i)$ is
assigned to  one of those  $n-c-1$ facilities $\frac{n-c-1}{c(n^2-1)}$
of the  time of those cases.  So in total  this happens $\frac{c}{n-1}
\times  \frac{n-c-1}{c(n^2-1)}  =  \frac{n-c-1}{(n-1)(n^2-1)}$ of  the
time, so it follows that client $j$ is assigned to a specific facility
$i',$ $i' \neq i,$ $i' \leq n-1,$ $\frac{n-c-1}{(n-1)(n-2)(n^2-1)}$ of the
time. The fraction with which  $j$ is assigned to $i'$ after $Round_A$
is $\frac{n-c-1}{(n-1)(n-2)(n^2-1)} \phi$.

For  the proof  of the  second part  of the  lemma,  consider facilities
$n,n+1$. Each one of those is present in the classes of type $A$ an equal 
fraction $1/2$ of the time. The
only clients that  are assigned to them are  their exclusive clients. Each
class $cl$ assigns exactly $B= n^2$ out of those $n^2+n-1$ clients. So,
due to symmetry, each client $j \in Exclusive(n)$ is present in
$cl$ $\frac{n^2}{n^2+n-1}$  of the time,  so $j$ is assigned to $n$ and $n+1$
with a fraction of $\frac{n^2}{2(n^2+n-1)} \phi$ to each.
\end{proof}

Note that  after $Round_A$ each  facility $i, i  \leq n-1,$ has  a total
amount $ \frac{(n-c-1)B}{(n-1)} \phi$  of clients (since it is present
in a class  $\frac{(n-c-1)}{(n-1)}$ of the time and  when this happens
it is  given $B$ clients).  Similarly, facilities $n,n+1$  after $Round_A$
have a total amount $B\phi /2$ each.

Now we can explain the underlying intuition for distinguishing between
the two
types of classes.  The feasible fractional solution $(y^*,x^*)$
we  intend to  construct is  the following:  for each
facility $ i\leq n-1,$ its exclusive clients are assigned to it with
a fraction of $\frac{n^2-1}{n^2}$ each, while they are assigned with a
fraction of $\frac{1}{(n^2)(n-2)}$ to  each other facility $ i'
\leq  n-1$. As  for facilities  $n,n+1$, all  of their exclusive  clients are
assigned with a fraction of $1/2$  to each.  If  we  project  the solution  to  
 $(y,x)$, the $y$ variables will be forced 
to  take   the  values
$y^*_i=\frac{n^2-1}{n^2},$ for $i \leq n-1,$ and $y^*_n=y^*_{n+1}=\frac{n^2+n-1}{2n^2}$. Observe as we give some  amount of  measure to  $Round_A$,
 the  variables  concerning the
assignments to facilities $n,n+1$ tend to their intended values in the
solution we want to construct ``faster'' than the variables concerning the
assignments to the other facilities. This is because, by Lemma~\ref{lemma:roundA}
after $Round_A$ each exclusive client  of $n,n+1$ is assigned to each of them with
a fraction of $\frac{n^2}{2(n^2+n-1)} \phi$ which is $\frac{n^2}{n^2+n-1}
\phi$ of  the intended value. At the  same time, every
exclusive  client of  each other  facility is  assigned to  it  with a
fraction of $\frac{n-c-1}{n-1} \phi$ which is $\frac{\frac{n-c-1}{n-1}
  \phi}{\frac{n^2-1}{n^2}}$  of the  intended value.  For sufficiently
large  instance  $I$,  as  $n$  tends  to  infinity,  the  assignments
to $n$ and $n+1$ will reach their intended values while there will
be   some    fraction   of   every    other   client   left    to   be
assigned. Subsequently we have to use classes of type $B$, 
to achieve the opposite effect: the
variables  concerning the  assignments of  the first  $n-1$ facilities
should tend
to their intended values ``faster''  than those of $n$ and $n+1$ (since
$n$ and $n+1$  are 
not  present in  any of  the classes  of type  $B$,  the corresponding
speed will actually be zero).

We  proceed  with  giving  the  details  of  the  usage  of  type  $B$
classes. As before,  we give to each such class  a very small quantity
of  measure $\epsilon$.  Let  $\xi$  be the  total  amount of  measure
used. We call this step $Round_B$.

\begin{lemma}   \label{lemma:roundB}
After  $Round_B$,  each client  $j  \in  Exclusive(i),$  $i\leq n-1,$  is
assigned  to $i$  with a  fraction of  $\frac{n-c}{n-1} \xi$  and is
assigned to each other  facility $i',$ $i' \neq i,$ $i'  \leq n-1,$ with a
fraction of $\frac{n-c}{(n-1)(n-2)(n^2-1)} \xi$.
\end{lemma}

\begin{proof}
The proof  follows closely that of  Lemma~\ref{lemma:roundA}. A  facility $i, i
\leq n-1,$  is present in  a class of  type $B$ $\frac{n-c}{n-1}$  of the
time (since  $c \geq  2$ this  fraction is less  than $1$).  Each such
time, every  client $j  \in Exclusive(i)$ is  assigned to  it (again
this  is due  to the  definition  of classes  of type  $B$). So  after
$Round_B$,   $j$   is  assigned   to   $i$   with   a  fraction   of
$\frac{n-c}{n-1} \xi$.  
Also, when $i$  is  present  in a  class, it is assigned exactly one client
which is exclusive to a facility
not in the class. Since in total there are $(n-2)(B-1)$ such candidate clients,
and by symmetry, after round $B$ 
each one of them is picked an equal fraction of the time to
be assigned to $i$, we have that
each client $j$ is assigned to a facility for which $j$ is not
exclusive with  a fraction
$\frac{n-c}{(n-1)(n-2)(n^2-1)} \xi$.
\end{proof}

\noindent
To  construct  the   aforementioned  fractional  solution $(y^*,x^*)$,  set  $\phi=
\frac{n^2+n-1}{n^2}$    and     $\xi=    (\frac{n^2-1}{n^2}-\frac{n-c-1}{n-1}\phi)
\frac{n-1}{n-c}$, and  add the  fractional assignments of  the two
rounds. 

It is easy to check that the facility and assignment variables of facilities $n,n+1$
take the value they have in $(y^*,x^*)$. Same is true for the facility variables for $i \leq n-1$
and the assignment variables of the clients to the facilities they are exclusive. 
To see that the same goes for the non-exclusive assignments, observe that since
every class assign exactly $B$ clients to its facilities we have that $\sum_j x_{ij}=By_i$.
So each $i\leq n-1$ takes exactly $1-1/n^2$ demand from non-exclusive clients which  are
$(n-2)(B-1)$ in total. Thus, by symmetry of the construction, each one them is assigned to $i$
with a fraction of $\frac{B-1}{n^2(n-2)(B-1)}=\frac{1}{n^2(n-2)}$

\subsection{Proof of unbounded integrality gap of the constructed solution}

In the  present subsection, we  manipulate the costs of  instance $I$,
which we left undefined, so as to create a large integrality gap while
ensuring that the distances form a metric.

Set each facility opening cost to zero. As for the connection costs (distances)
consider the $(n-2)$-dimensional Euclidean space $\mathbb{R}^{n-2}$. Put
every facility $i,$  $i\leq n-1,$ together with its  exclusive clients on a
distinct vertex of an $(n-2)$-dimensional regular simplex with edge length
$D$. Put facilities $n,n+1$ together with their exclusive clients to a point
far away  from the simplex, so  the minimum distance from  a vertex is
$D' >> D.$ Setting $D'=\Omega(nD)$ is enough.

Since the distance between a facility and one of its exclusive clients
is  $0$,  the  cost  of  the fractional  solution  we  constructed  is
$O(nD)$. This cost  is due to the assignments  of exclusive clients of
facility $i,$ $i \leq n-1,$ to facilities $i'$ with $i' \neq i,$ $i' \leq n-1.$ 
As  for the cost  of an arbitrary integral  solution, observe
that since the $n^2+n-1$ exclusive  clients of $n,n+1$ are very far from
the  rest of  the facilities,  using $n$  of them  to  satisfy some
demand of  those facilities and help  to open all of  them, incurs a
cost of $\Omega(nD') = \Omega(n^2D).$ On the other hand, if we do not open all of the $n-1$
facilities on  the vertices of the  simplex (since they  have in total
$(n-1)(B-1)$  exclusive clients  which is  not enough  to open  all of
them), there  must be  at least  one such facility  not opened  in the
solution, thus its $B-1=\Theta(n^2)$ exclusive clients must be assigned elsewhere,
incurring a cost of  $\Omega(n^2D).$

This concludes the proof of Theorem~\ref{theorem:proper}.

\section{Proof of Theorem ~\ref{theorem:proper2}} 
\label{sec:proof_theorem_p2}

The proof is similar to that for \lbfl. 
We prove that the relaxation must use 
a specific set of classes and then we use these classes to construct a
desired feasible solution. In the last step we 
 define appropriately  the costs of the instance. 

\subsection{Existence of a specific type of classes}

Consider
an instance $I$ with $n$ facilities, where $n$ is
sufficiently large to ensure  that $\alpha n \leq n - c_0$  where
$c_0,$  is a
constant greater than or equal to  $1$. Let the capacity be $U=n^2$, and let
the number of  clients be $(n-1)U+1$. Notice that in every integer solution of the instance
 we must open  at least $n$ facilities. The  facility costs  and the assignment  costs will  be defined
later.  

 We  assume, like before, that  the  facilities are  numbered
$1,2,...,n$. 
Consider  an  integral  solution  $s$  for $I$   where  all the
facilities  are opened, and furthermore 
facilities $1..., n-1$  are assigned $U$ clients each 
and facility $n$ is assigned one client. 
Since our proper  relaxation is valid, there must be a solution $s'$ in the  space of
feasible solutions of the proper relaxation whose $(y,x)$ projection is the characteristic 
vector of $s$.  
By Definition~\ref{def:constell},
it is easy to see that $s'$ can only be obtained as a 
positive combination of classes $cl$ such that for
every    facility   $i$    we   have    $Clients_{cl}(i)   \subseteq
Clients_s(i)$.  Recall  that  since  the  complexity  of  our
relaxation is $\alpha$, the classes in the support of any solution 
have at most $n-c_0 \leq n-1$
facilities. 

Now consider the support  of $s'$. We will distinguish the classes $cl$ for
which variable $x_{cl}$ is in the support of $s'$ into 2 sets. The first set consists 
of the classes that assign exactly one client to facility $n$; call them \emph{type A} classes.
The second set  consists  of the classes that do not assign any client to facility 
$n$; call those \emph{type B} classes. By the discussion above those sets form a
partition of the classes in the support of $s'$, and moreover they are both non-empty: this is
 by the fact that at most  $n-c_0$ facilities are in any class, and by the fact
 that in $s$ all $n$ 
facilities are opened integrally. Notice also that no  class
of type B can contain facility $n$ even though the definition of a class does not
exclude the possibility that a class contains a facility to which no clients are
assigned. 

We call \emph{density} of  a class $cl$ the ratio 
$d(cl)=\frac{\sum_{i\neq n}|Clients_{cl}(i)|}{|F(cl)-\{n\}|}$. By the discussion 
above we have that $d(cl)\leq U$ for all $cl$ in the support of $s'$. The following holds:

\begin{lemma}
All classes in the support of $s'$ have density $U.$
\end{lemma}

\begin{proof}
The amount of demand that a class $cl$ contributes to the demand assigned to the set
of the first $n-1$ facilities by $s'$ is $d(cl)|F(cl)-\{n\}|x_{cl}.$
 We have $\sum_{cl}d(cl)|F(cl)-\{n\}|x_{cl}=(n-1)U$.
 Observe
that by the projection of $s'$ on $(y,x)$ and by the fact that for $i=1,...,n-1$,
 $y_i=1$ in $s$, we have $\sum_{cl}|F(cl)-\{n\}|x_{cl}=n-1$. 
Setting $m_{cl}=\frac{x_{cl}|F(cl)-\{n\}|}{n-1}$ we have from
the above $\sum_{cl}m_{cl}=1$ and $\sum_{cl}m_{cl}d(cl)=U$. 
The latter together with the fact that $d(cl)\leq U$ we have that $d(cl)=U$ for all classes
$cl$ in the support of $s'$.
\end{proof}

The following corollary is immediate from the above:

\begin{corollary}
There is a type $B$ class in the support of $s'$ that has density $U.$
\end{corollary}

So far we have proved that 
in the class set of any proper relaxation for $I,$ there is 
a class $cl_0$ of type $B$ with density $d(cl_0)=C$.
 Let $|F(cl_0)| = t \leq n-1.$

\subsection{Construction of a bad solution}

Consider the symmetric classes of $cl_0$ for all permutations of the $n$ facilities
and for all permutations of the clients. Those classes are not necessarily in the support of $s'$. Take a quantity of measure $\epsilon$ and distribute it equally among all 
those classes. Since class $cl_0$ has density $U,$ all those symmetric classes
assign on average $U$ clients to each of their facilities. 
Due to symmetry, each facility is in a class $\epsilon \frac{t}{n}$ of the time and is assigned $\epsilon \frac{t}{n}U$ demand. Each client is assigned to
each facility $\epsilon\frac{tU}{((n-1)U+1)n}$ of the time. We call that step of our construction \emph{round $A$}.

Now consider the symmetric classes of $cl_0$ for all permutations of the first $n-1$ facilities
and for all permutations of the clients (those classes are well defined since $t\leq n-1$).
Again distribute a quantity of measure $\epsilon$ equally among all 
those classes. Similarly to the previous, each facility is in a class $\epsilon \frac{t}{n-1}$ of the time and is assigned $\epsilon \frac{t}{n-1}U$ demand. Each client is assigned to
each facility $\epsilon\frac{tU}{((n-1)U+1)(n-1)}$ of the time. 
We call that step of our construction \emph{round $B$}.

Spending $\phi=\frac{1}{nt}$ measure in round $A$ and $\xi=\frac{(n-1)(1-1/n^2)}{t}$ 
measure in round $B$ we construct a solution $s_b$ whose projection to $(y,x)$ is the 
following $(y^*,x^*)$:
$y^*_i=1$ for $i=1,...,n-1$, $y^*_n=\frac{1}{n^2}$, and for every client $j,$ $x^*_{nj}=\frac{U/n^2}{(n-1)U+1}$ and 
$x^*_{ij}=\frac{1-x^*_{nj}}{n-1}$ for $i=1,...,n-1.$ It is easy to see that $s_b$ is
a feasible solution for our proper relaxation.

Now simply set all distances to $0$, and define the facility opening costs as 
$f_n=1$ and $f_i=0$ for $i\leq n-1.$ It is easy to see
that the integrality gap of the proper relaxation is $\Omega (n^2)$. 
In Section \ref{sec:ls_cfl}, where we prove unbounded integrality gap for the
Lov\'{a}sz-Schrijver procedure, 
 we will have to use a somewhat more general "bad" solution on an instance with 
 many costly facilities.

\section{The LS Hierarchy}
\label{sec:ls}

The Lov\'{a}sz-Schrijver hierarchy was defined in \cite{LovaszS91}. For a 
comprehensive presentation and 
various reformulations see \cite{Tulsiani11-chapter}. In this section
we give the necessary definitions and the reformulation we are
employing in our proof. 

In \cite{LovaszS91} an operator $N$ was defined which refines a convex set 
$K \subseteq [0,1]^n$, when applied to  it.
After $n$ applications the resulting convex set is the integer hull of $K$. 
Starting with a polytope $P \subseteq [0,1]^n$ we define
$cone(P)=\{y = (\lambda, \lambda z_1,... , \lambda z_n) \mid  \lambda \geq 0, (z_1,...,z_n) \in P\}$.
The following Lemma characterizes the vectors  of $cone(P)$ 
that survive  after $m$ iterations.

\begin{lemma} [\cite{LovaszS91}]\label{lemmals}
 If $K$ is a cone in $\mathbb{R}^{n+1}$, then $z \in N^m(K)$ iff there is an $(n+1)\times (n+1)$ symmetric
matrix $Y$ satisfying
\begin{itemize}
\item[1.] $Ye_0=diag(Y)=z.$
\item[2.] For $1 \leq i \leq n$, both $Ye_i$ and $Y(e_0-e_i)$ are in $N^{m-1}(K).$
\end{itemize}
\end{lemma}

In such a case, $Y$ is called \emph{the protection matrix} of $z$. Since we are interested in
the projection of the cones on the hyperplane $z_0=1$ which contains our
original polytope, we restate the conditions of survival of $z$ as the following
corollary which is immediate from Lemma \ref{lemmals}.

\begin{corollary} [\cite{AroraBL02}]\label{corls}
 Let $K$ be a cone in $\mathbb{R}^{n+1}$ and suppose $z \in \mathbb{R}^{n+1}$ 
where $z_0 = 1.$ Then $z \in N^m(K)$ iff there
is an $(n+1)\times(n+1)$ symmetric matrix $Y$ satisfying

\begin{itemize}
\item[1.]  $Ye_0= diag(Y) = z.$
\item[2.]  For $1 \leq i \leq n$: If $z_i = 0$ then $Ye_i ={\bf 0}$; 
If $z_i =1$ then $Ye_i=z$; Otherwise, $Ye_i/z_i$, $Y(e_0-e_i)/(1-z_i)$
both lie in the projection of $N^{m-1}(K)$ onto the hyperplane $z_0 = 1.$
\end{itemize}
\end{corollary}

Let $Y_i$ denote the vector $Y^Te_i,$ i.e., the $i$th row of $Y.$ 
Corollary  \ref{corls} makes  it  convenient to  work with  individual
vector solutions that can be  combined as rows to build the protection
matrix. 
 We focus now on the survival of
a vector  $z$ for one  round and state  some simple properties of $Y.$

Given a protection matrix $Y$ of $z,$ 
we define a set of at most $2n$ {\em
  witnesses} of vector $z.$ For each variable $z_i$, $1\leq i \leq n,$ there are at most
$2$ such witnesses: the one that equals $Y_{{i}}/z_{i}$
(if  $z_i\neq 0$),   which we  call \emph{type  1 witness  of $z$
  corresponding to variable $z_i,$} and the vector $\frac{Y_0-Y_{i}}{1-z_i}$
(if $z_i \neq 1$), which we call \emph{type 2 witness of $z$
  corresponding to  variable $z_i.$} For the validity  of the upcoming
observation recall  that if $z_i =0,$  and hence the  type $1$ witness
corresponding to $i$ is undefined, $Y_i = {\bf 0}.$ 

\begin{observation}\label{obs_diag}
The condition that $Y$'s main diagonal is equal to the vector $Y_0$ is
equivalent to the following: the variable $z'_i$ of the type $1$ witness
$z'$ corresponding to variable $z_i\neq 0$ is equal to $1.$
\end{observation}

The rows of $Y$ that correspond to zero variables in $z$ are filled 
with zeros and called {\em
  special.}  
Moreover if $z_i=1,$ $Y_i= z.$ 
To account for these requirements it is not enough that the integer values
in $Y_0$ appear on the main diagonal. The following claim states that
they are replicated across  all witnesses. 

\begin{claim}\label{simplefact1}
Let  $z'$ be a witness  of  $z$.  If  for some $i$, $z_i\in
\{0,1\}$, then $z'_i=z_i.$
\end{claim}

To enforce symmetry  for a special row $Y_i={\bf  0}$ that corresponds
to a  variable $z_i=0,$ it must be  the case that the  $i$th column is
set to  zero as well.  This is ensured by  Claim~\ref{simplefact1} for
all  entries $Y_{ki}$  of  the column  for  which $z_k  \neq 0.$  (The
remaining entries  of the column belong  to special zero  rows and are
equal to zero anyway). 
For  the remaining rows, it will be convenient to 
  express  each variable of a type 1 $z'$  child of $z$ corresponding
to some variable  $z_i$, by defining the factors by which  the variables of $z'$
differ from the  corresponding variables of $z$. Then  the symmetry condition of
$Y$ is satisfied by ensuring that  the condition of the following claim on those
factors holds.

\begin{claim}\label{simplefact2}
Let indices $q, t$ take values in $\{1,\ldots,n\}.$ The symmetry condition of the
protection  matrix  of  $z$  holds  iff  Claim~\ref{simplefact1} holds
 and for  each type  1  witness   $z'$  of  $z$
corresponding to variable $z_q$, for which 
$z'_t=z_tf$, $z_t \neq 0,$ 
then, for  the type 1 witness  $z''$ of  $z$ corresponding to variable
$z_t$, we have $z''_q=z_qf$.
\end{claim}

Observe that  when we construct a  type $1$ witness $z'$  corresponding to $z_t$,
 the type  $2$ witness  $z''$ corresponding  to $z_t$ is  automatically defined. We
 say that $z''$ is the \emph{twin} of $z'$ .
 
\begin{claim}\label{simplefact3}
Let indices $q, t$ take values in $\{1,\ldots,n\}.$ If the protection matrix of $z$
exists, 
the following must hold. 
If $z'_q=z_q(1+\epsilon)$ in the type $1$ witness  corresponding to
$z_t,$ $z_t \neq 1,$ then
$z''_q=z_q(1-\frac{z_t\epsilon}{1-z_t})$ where $z''$ is the  type $2$ twin of $z'.$
\end{claim}

To prove the existence of a  protection matrix $Y$ for a vector $z$ we
will proceed as follows. We will define a set $S(z)$ of witness
vectors that  will contain a type $1$  and a type $2$  witness for every
non-integer variable and  one of the appropriate type  for each 
integer  variable. We  will ensure  that the  vectors in  $S(z)$  meet the
conditions of Observation~\ref{obs_diag}, Claims~\ref{simplefact1},
\ref{simplefact2}  and \ref{simplefact3}. In  addition  we will prove 
that the vectors
meet the  feasibility constraints for $K.$ Then 
we will have shown how to  construct  a protection matrix $Y$ of $z$ whose rows
consist of: 
the type  1 vectors  from $S(z)$ scaled each  by the  corresponding variable
$z_i,$ together with  one special {\bf  0} vector  for each zero variable in $z.$

To prove  the survival of a vector  for many rounds we  just embed the
strategy above in an inductive argument. 
The following fact  is immediate from Corollary
\ref{corls}:   if,   for   all   $i$   the   vectors   $Y_i/z_i$   and
$\frac{Y_0-Y_i}{1-z_i}$ (that are  defined) witnessing the survival of
our initial vector $z$, survive  themselves $k$ rounds of LS, then $z$
survives actually $k+1$ rounds of LS.

We
define a tree structure which we 
call the  \emph{evolution tree $T_z$ of $z.$} Every node in $T_z$ is {\em
  associated} with a vector. The tree is defined 
recursively: vector $z$ is associated with the root of the tree, and if $v$ is a node
of $T_z,$ associated with vector $z(v),$  then  the vectors  witnessing the
survival of $z(v)$ 
are associated in one-to-one manner with the children 
of $v.$ If  there are no such witnesses,  the fractional solution $z(v)$  does
not survive one round of LS, and  we call $v$ a \emph{terminal} node. 
The number  of rounds  that our
initial vector $z$ survives, is the length of the shortest path from the root of
the evolution tree $T_z$  to a terminal node.

Given a root vector, we will show that as long as we have walked down the
evolution tree at depth $k$, where $k$  is the target number of rounds, then the
protection matrices of the root and all its descendants are well-defined. 
The inductive step shows how to define all children of a node $v$ and
therefore increase the depth of the tree by one. 
From now on  we  refer
interchangeably to a node and its  associated solution
vector. Accordingly, if $v'$ is a child of node $v,$ $z'$ ($z$) is 
associated with $v'$ (resp. $v$) and 
and  $z'$ is a type $1$ ($2$) witness  of $z$ corresponding to variable
$z_i,$ 
we will refer to node $v'$ as 
a {\em type $1$ (\mbox{resp.} $2$)  child of node-solution $z$
  corresponding to variable $z_i.$}

Finally, the following fact will be useful for the feasibility proof.

\begin{lemma}\label{eqconlemma}
Given a solution $z$ in  the evolution tree that satisfies an equality
constraint $\sum_i a_{i}z_{i}=b$,  and given a child of  $z$ that is a
type  $1$ solution  $z'$ corresponding  to some  $z_t$  that satisfies
$\sum_i a_{i}z'_{i}=b$, then the twin  type $2$ solution $z''$ of $z'$
also satisfies $\sum_i a_{i}z''_{i}=b$.
\end{lemma}
\begin{proof}
Let  $z'_{i}=z_{i}(1+\epsilon_i)$.   From  $\sum_i  a_{i}z_{i}=b$  and
$\sum_i a_{i}z'_{i}=b$ we get $\sum_i a_{i}z_{i}\epsilon_i=0$. Then by
Claim~\ref{simplefact3},          $\sum_i          a_{i}z''_{i}=\sum_i
a_{i}z_{i}(1-\frac{z_t\epsilon_i}{1-z_t})=\sum_ia_{i}z_{i}            -
\frac{z_t}{1-z_t}\sum_i a_{i}z_{i}\epsilon_i=b.$
\end{proof}

\section{Proof of Theorem~\ref{theorem:ls-cfl}}
\label{sec:ls_cfl}

In  this  section we  show  that the  integrality  gap  on a  suitable
instance of \cfl\ remains unbounded even after applying a large number
of iterations of the LS procedure.  

The instance  is the following.  Consider a set  of $n$
facilities  which have $0$  opening cost.  We call  that set
\emph{Cheap}. Moreover, 
consider a  set of $l$ facilities  that have an opening cost  of $1$ each.
Call that set  \emph{Costly}. Think of $l$ as being $\Theta(n)$; we will later
prove that the number of rounds  of survival are maximized for
$l=n$. The set of facilities $F$ is $Cheap\cup
Costly.$  Let all the facilities have the same capacity
$U=n^3$, and let there be a  total of $nU+1$ clients in the set $C.$  All clients and facilities
are at a  distance of $0$ of  each other. Clearly all integral  solutions to the
instance have a cost of at least $1$.

Consider  the following solution  $s$ to  (LP-classic): For  each
facility   $i   \in   Cheap,$   $y_i=1$,    and   for   each   client   $j,$   set
$x_{ij}=\frac{(1-a)}{n}$,  $a=n^{-2}$.  For  each  facility  $i\in  Costly,$  
$y_i=H/n^{2}=b$, for  some sufficiently large 
constant $H$ ($H=10$ is enough), and  for each
client  $j,$  set  $x_{ij}=a/l$.  The  constructed  solution  incurs  a  cost  of
$\frac{l H}{n^{2}},$ which is $\Theta(n^{-1})$ if  $l=\Theta(n).$

It is well-known that some simple valid 
inequalities are not produced early in the LS procedure.
For example, in the case of \cfl\ 
our proof implies that $\Theta(n)$ rounds
are required to obtain the simple inequality $\sum_{i \in F}  y_i \geq \lceil
|{C}| / U  \rceil$ which is facet-inducing for our instance. 
This inequality is not critical however for our proof. 
It is easy to modify the input by adding one  facility and
one client at a large distance from the rest of the instance, 
so that Theorem~\ref{theorem:ls-cfl} continues to hold 
while the inequality   above is  satisfied by a bad fractional solution.
Given an analogous  fixed set of inequalities, an adversary can modify the
instance in a similar manner.

Solution  $s$ cannot  survive $l$  rounds of  application of  the  LS procedure:
consider the path  from the root where we  descend each time via a  type 2 child
corresponding to  a $y_i$  variable of a  costly facility (a  different facility
each time).  After $l$ such steps, assuming  of course that the  nodes are defined,
one can show  that all facilities in $Costly$ will be  closed and the facilities
in $Cheap$  have to absorb  all the  demand in the  instance, which leads  to an
infeasible solution.

\begin{observation}
\label{observation:l-rounds}
Solution $s$ survives less than $l$ rounds of the LS procedure.
\end{observation}

For the proof of Observation \ref{observation:l-rounds}
we will actually prove the following stronger Lemma:

\begin{lemma}
Let $S$ be a set of variables $z_1,z_2,...,z_t$ of vector $s_0,$ s.t. $\sum_{z_i \in S}z_i < 1.$
 If  $s_0$ survives $|S|$ rounds of LS, then 
 there is a path $p$ of length $|S|$ starting from the root $s^0$ of the evolution tree 
that ends with a node-solution $s^{(|S|)},$ such that in $s^{(|S|)},$ for all $z_i
\in S,$ $z_{i}=0.$ 
\end{lemma}

\begin{proof}
The proof is by induction on $|S|$.
Suppose that $s^0$ survive $|S|$ rounds. Let $z_{j}$ be the variable in  $S$ with
the highest value. Consider the type 2 $s'$ child of $s^{0}$ corresponding to $z_{j}$.
Then $z'_{i} \leq z_{i}(1+\frac{z_j}{1-z_j})$ otherwise by Claim
$\ref{simplefact3}$ 
$z''_i< 0$ in the twin type $1$ solution
$s''$ of $s'.$ So we have \\
$\sum_{z_i \in S-\{j\}}z'_{i} \leq $\\
$\sum_{z_i \in S-\{j\}}z_i+\sum_{z_i \in S-\{j\}}z_i\frac{z_j}{1-z_j} <$\\
$1-z_j + z_j =1$.\\

\noindent
Setting $S'=S-\{j\}$ we have $\sum_{z_i \in S'}z_i< 1.$ By the  inductive hypothesis
the evolution tree contains a  path $p'$ starting at $s'$ that has 
length $|S|-1.$  By appending  $s^{0}$ before the first node $s'$ of $p'$
we obtain the desired path $p$.
\end{proof}

To prove Observation \ref{observation:l-rounds}, assume that $s^0$ survives $l$ rounds.
Then there is a path of length $l$ starting from the root of the evolution tree,
such that in  the last
node solution $(y,x)^{(l)}$ of the path all the facilities in 
$Costly$ are closed. Clearly this cannot
be a feasible solution.

\noindent 
We are ready to state the main theorem of this section 
which implies that the solution $s$  survives $l/10$ 
 rounds of LS. We  do not make any attempt to optimize the   constant.
At every level of the induction the new witness solutions cannot
differ drastically from  their parent node. 
We identify a  set of invariants that express this controlled
evolution of the values.

\begin{theorem}\label{smalltheorem}
Let $l \leq n$ and $\delta$ be a constant of value $1/H.$ 
 We can construct an evolution tree $T_s$ with root $s$
such that
any  node  $u$ of $T_s$ at depth $k \leq \frac{l}{10}$ 
is associated with a  feasible solution $(y,x)$ that satisfies the following invariants:\\
\vspace*{-0.8cm}
\begin{itemize}
\item[1] For variable $y_i\notin \{0,1\},$ $i \in Costly$, $y_i \geq b-2k\frac{a}{l}$
 and $y_{i}\leq b+2k\frac{a}{l}.$
\item[2] 
\begin{itemize}
\item[(a)] For variable $x_{ij} \notin \{0,1\}$, $i \in Cheap,$
  $\frac{1-a}{n}-2k\frac{a}{nl}b^{-1} \leq  x_{ij} \geq \frac{1-a}{n}+2k\frac{1-a}{n}\max\{1/l,1/n\}.$
\item[(b)] For variable $x_{ij} \neq 0$, $i \in Costly,$ and $y_i \notin \{0,1\}$, $\frac{a}{l} \leq x_{ij}\leq \frac{a}{l}+2k\frac{a(1-a)}{nl}$.
\item[(c)]  For variable $x_{ij} \notin \{0,1\}$, $i \in Costly,$ and
  $y_i=1$, $\frac{a}{l} \leq x_{ij}\leq
  (\frac{a}{l}+2k\frac{a(1-a)}{nl})b^{-1}(1+\nolinebreak \delta)$.
\end{itemize}
\item[3] For $i \in Cheap$, $\sum_{j}x_{ij} \leq (nU+1)\frac{1-a}{n}+2k(nU+1)\frac{a}{nl}$.
\item[4] For $i \in Costly$,
\begin{itemize}
\item[(a)]  if $y_i \neq 1$,$\sum_j x_{ij}\leq (nU+1)\frac{a}{l}+k$.
\item[(b)]  if $y_i = 1$,$\sum_j x_{ij}\leq ((nU+1)\frac{a}{l}+k)(1+\delta)b^{-1}$.
\end{itemize}

\end{itemize}

\end{theorem}

Setting in our instance $l=n,$ by Theorem~\ref{smalltheorem} we obtain that the solution
$s$ survives $\Omega(n)$ rounds. 
Thus we have proved Theorem \ref{theorem:ls-cfl}.

\subsection{Proof of Theorem ~\ref{smalltheorem}}
\label{subsec:invproof}

The  proof is  by induction  on the  depth of  node $u$.  More  specifically, by
assuming that the invariants hold for an arbitrary node $v$ at depth
less than $l/10$, we
show  how  to construct  all  the  children nodes  of  $v$  so   that they  are
well-defined and 
the invariants are met.  

In the proof, whenever we give the  construction of a type $1$ or type $2$ child
of  $v$  corresponding  to  some  variable  $z_i$, we  refer  to  $z_i$  as  the
\emph{touched  variable} --  we also  say that  $z_i$ is
\emph{touched}  as type 1 or type 2 in the
current step.  We will consider cases according to which variable is touched and
whether it  is touched  as type $1$  or as  type $2$. When  we touch  a variable
$z_i\notin  \{0,1\}$  as type  $1$,  $z_i$  always takes  the  value  $1$ so  by
Observation \ref{obs_diag}  we satisfy  the condition that  the diagonal  of the
underlying protection matrix is equal to the $0$th row.  Note that we will not  give
the construction  for the case in  which $y_i,$ $i\in Cheap,$  is touched, since
$y_i$ is  always $1$ and  the construction is  trivial in those cases.  The same
applies  to  the  cases of  all  variables  that  have  integral values  in  the
node-solution $v$ of the inductive hypothesis, as we simply 
enforce Claim~\ref{simplefact1}.
 
Another   feature of   our construction is  the following: when  a fractional
variable $x_{i'j}$ is touched as type $1$, it is set to  $1,$ and for all
$i \neq i',$  $x_{ij}$
becomes $0$. If $x_{i'j}$ is touched as type $2$, it is set to  $0$ and
in order to maintain feasibility 
its previous value is
distributed  among the  other assignments  of client  $j$. Thus  for  every $j$,
either there is  some $i'$ such that  $x_{i'j}=1$ and for all other  $i \neq i',$
$x_{ij}=0$ (e.g.,  when cases $1b$,  $1c$ below have happened for  an ancestor of  $v$), or
there are at most $k$ facilities to which the assignment of $j$ is $0$ (if there
are type $2$ nodes,  through cases $2a$, $2b$, $2c$, along the  path of the tree
that leads  to $v$). 
In fact, as far as assignments to cheap facilities are concerned,
the upper bound of $k$ holds cumulatively  across all clients,
since 
no more than  $k$ assignment variables can be touched as Type 2 along  a
path of length $k.$ 
Specifically, 
let $C'$ be the set of clients $j$ for which, for all $i \in F,$ $x_{ij} < 1.$ 
We will use the fact that $| \{ x_{ij},  i \in Cheap, j \in C' \mid
x_{ij}=0 \}| <
k.$

Note  that  the  invariants of    Theorem~\ref{smalltheorem}  imply the  satisfaction  of
constraints  \eqref{x<y},\eqref{1>x>0},\eqref{eq:nni} and  \eqref{sat1}  for the
number  of  rounds  we consider.  Thus,  when  proving  the feasibility  of  the
constructed solution each time, we only have to ensure that \eqref{eq} holds.

\begin{lemma}\label{easyconstraints2}
Let  $(y,x)$  be a node-solution defined  at   depth  $k\leq
\frac{l}{10}$ of the evolution tree $T_s.$  If $(y,x)$ satisfies
Invariants 1--4, then $(y,x)$ meets 
constraints \eqref{x<y},\eqref{1>x>0},\eqref{eq:nni} and \eqref{sat1}.
\end{lemma}

We now explain the  inductive step that constructs the children of
node $v,$ where $v$ is at depth $k < l/10.$ We  distinguish  cases according to the
variable that is touched. 

\medskip
\noindent
{\bf Case 1: type $1$ children}

\noindent 
{\bf \underline{subcase $1a$:} touched variable is  $y_{i_k}$, $i_k \in Costly$}

{\sc Algorithm}

Consider the type $1$ child $(y',x')$ of $v$ corresponding to variable $y_{i_k}$. Variables $y_{i_k},x_{i_kj}$ for all $j$ are multiplied by a factor of $1/y_{i_k}$ and so  $y'_{i_k}=1$.
Note that since we only consider cases where $y_{i_k}$ is fractional, by 
the inductive hypothesis we have that for all variables $x_{i_kj}$, Invariant 2.b holds.
The variables involving facilities $i' \in Costly-\{i_k\}$, namely $y_{i'},x_{i'j}$ for all $j$, remain the same.  For all $j$ and for all $i' \in Cheap$ such that $x_{i'j}\neq 0$  we have $x'_{i'j}=x_{i'j}-\frac{(1/y_{i_k}-1)x_{i_kj}}{t}=x_{i'j}(1-\frac{(1/y_{i_k}-1)x_{i_kj}}{x_{i'j}t})$, where $t$ is
the number of facilities in $Cheap$ for which $j$ is assigned with a non-zero fraction (so $t \geq n-k$).

{\sc Feasibility}

Constraint \eqref{eq} is satisfied by construction:\\

\noindent
$\sum_i x'_{ij}=\sum_i x_{ij} +(1/y_{i_k}-1)x_{i_kj}-\sum_{i \in Cheap| x_{ij}>0}\frac{(1/y_{i_k}-1)x_{i_kj}}{t}=$\\
$\sum_i x_{ij}=1$\\

{\sc Invariants}

{\sf Invariant $1$}

For $i \in Costly-\{i_k\}$, $y_i$ remain unchanged so Invariant $1$ holds by
the inductive hypothesis (from now abbreviated as i.h).

{\sf Invariant $2$}

For $i \in Cheap$ we have $2.a$:

\noindent
$x'_{ij}=x_{ij}-\frac{(1/y_{i_k}-1)x_{i_kj}}{t}\geq$ \hfill  (by Invariants $1$, $2$ of i.h. and being  generous)\\
$\frac{1-a}{n}-2k\frac{a}{nl}b^{-1}-2b^{-1}\frac{a}{nl}\geq$\\
$\frac{1-a}{n}-2(k+1)\frac{a}{nl}b^{-1}$\\

For $i \in Costly-\{i_k\}$, $2.b$ holds since variables $x_{ij}$ were not changed. For $x_{i_kj}$:

\noindent
$x'_{i_kj}=x_{i_kj}\frac{1}{y_{i_k}}\leq$ \hfill (by Invariants $2.b$, $1$)\\
$(\frac{a}{l}+2k\frac{a(1-a)}{nl})b^{-1}(1+o(1))$

{\sf Invariant $3$}

Observe than the total demand assigned to each facility in $Cheap$ was decreased so Invariant $3$ holds by the 
inductive hypothesis.

{\sf Invariant $4$}

For $i \in Costly-\{i_k\}$ Invariant 4 holds by inductive hypothesis. For $i_k$ we have $4.b$:

\noindent
$\sum_j x'_{i_kj}=1/y_{i_k}\sum_j x_{i_kj}\leq$ \hfill  (by the invariants of i.h.)\\
$b^{-1}(1+o(1))((nU+1)\frac{a}{l}+k)$\\

\medskip
\noindent
  {\bf \underline{subcase $1b$:} touched variable is  $x_{i_kj*}$, $i_k \in Costly$}

{\sc Algorithm}

Consider the type $1$ children $(y',x')$ of $v$ corresponding to variable $x_{i_kj*}$. Variable $y_{i_k}$ is multiplied by a factor of $1/y_{i_k}$ and so  $y'_{i_k}=1$ (and  of course $x'_{i_kj*}=1$, and $x'_{ij*}=0$ for $i\neq i_k$). Every other variable remains the same.

{\sc Feasibility}
The feasibility of this case is trivial.

{\sc Invariants}
The Invariants $1,$ $2,$ $3$ in this case are satisfied trivially. For $4$ we have for facility $i_k$:

\noindent
$\sum_j x'_{i_kj}\leq$\hfill (variable $x_{i_kj*}$ becomes $1$)\\
$ \sum_j x_{i_kj} + 1\leq$\hfill (by $4$ of i.h.)\\
$(nU+1)\frac{a}{l}+k+1$\hfill  if $y_{i_k} \neq 1$  or \\
$((nU+1)\frac{a}{l}+k+1)b^{-1}$\hfill if $y_{i_k} =1$\\ 

In either of the two cases Invariant $4.b$ holds for the new value $y'_{i_k}.$

\medskip
\noindent
{\bf \underline{subcase  $1c$:} touched variable is  $x_{i_kj*}$, $i_k \in Cheap$}

{\sc Algorithm}

Consider the type $1$ children $(y',x')$ of $v$ corresponding to variable $x_{i_kj*}$. Variables $y_{i}, i \in Costly$ with $y_{i}\notin \{0,1\}$ are multiplied by a factor of $(1-\frac{(1/y_{i}-1)x_{ij*}}{x_{i_kj*}t})$, where $t$ is again the number of
facilities in $Cheap$ for which $j*$ is assigned with a non zero fraction (so $t \geq n-k$).
Of course $x'_{i_kj*}=1$, and $x'_{ij*}=0$ for $i\neq i_k$ as usual. Every other variable remains the same.

{\sc Feasibility}
Obviously \eqref{eq} is satisfied. All other constraints are satisfied by 
Lemma \ref{easyconstraints2}.

{\sc Invariants}

{\sf Invariant $1$}

For each $i \in Costly$ such that $y_{i}\notin \{0,1\}$ we have:

\noindent
$y'_i=y_i(1-\frac{(1/y_{i}-1)x_{ij*}}{x_{i_kj*}t})\geq$\hfill (by Invariant $1$ of i.h.)\\
$b-2k\frac{a}{l}
-y_i(\frac{(1/y_{i}-1)x_{ij*}}{x_{i_kj*}t})\geq$\hfill (by Invariants
$1$, $2.b$ of i.h.)\\
$b-2k\frac{a}{l}- 2\frac{a}{l} = b-(2k+2)\frac{a}{l}$

{\sf Invariant $2$}

Variables $x'_{ij}$ remain unchanged for $j\neq j*$. For $j*$, $x'_{i_kj*}=1$ while for $i\neq i_k$
we have $x'_{ij*}=0$, so $2$ is trivially satisfied.

{\sf Invariant $3$}

For $i\in Cheap-\{i_k\}$ the total demand is decreased (because of $j*$). For $i_k$:

\noindent
$\sum_j x'_{i_kj}\leq \sum_j x_{i_kj} +1\leq$\hfill (by $3$ of i.h.)\\
$(nU+1)\frac{1-a}{n}+2k(nU+1)\frac{a}{nl}+1\leq (nU+1)\frac{1-a}{n}+2(k+1)(nU+1)\frac{a}{nl}$\\

{\sf Invariant $4$}

The demand assigned to facilities in $Costly-\{i_k\}$ is decreased  (because of $j*$) 
so $4.a,$ $4.b$ trivially hold.

\medskip \medskip
\noindent
{\bf Case 2: type $2$ children}

\noindent 
{\bf \underline{subcase $2a$:} touched variable is  $y_{i_k}$, $i_k \in Costly$}

{\sc Algorithm}

Consider the type $2$ children $(y',x')$ of $v$ corresponding to variable $y_{i_k}\notin \{0,1\}$.
Let $f=\frac{y_{i_k}}{1-y_{i_k}}$. Solution $(y',x')$ is dictated by its twin type $1$ solution (case 1a):
 variables $y_{i_k},x_{i_kj}$ for all $j,$ are multiplied by a factor of $(1-f(1/y_{i_k}-1))$ and so  $y'_{i_k}=0$ and $x'_{i_kj}=0$, that is facility $i_k$ closes. The variables involving facilities $i' \in Costly-\{i_k\}$, namely $y_{i'},x_{i'j}$ for all $j$, remain the same. 
For all $j$ and for $i' \in Cheap$ such that $x_{i'j}\neq 0$ we have $x'_{i'j}=x_{i'j}(1+\frac{f(1/y_{i_k}-1)x_{i_kj}}{x_{i'j}t})$, where $t$ is again the number of
facilities in $Cheap$ for which $j$ is assigned with a non zero fraction (so $t \geq n-k$).

{\sc Feasibility}
Constraint \eqref{eq} is satisfied by Lemma \ref{eqconlemma}.

{\sc Invariants}

{\sf Invariant $1$}

For $i \in Costly-\{i_k\}$, $y_i$ remain unchanged so Invariant $1$ holds by inductive hypothesis.

{\sf Invariant $2$}

For $i \in Cheap$ we have $2.a$:\\

\noindent
$x'_{ij}=x_{i'j}+\frac{f(1/y_{i_k}-1)x_{i_kj}}{t}\leq$ \hfill (by Invariants $1$, $2$ of i.h.)\\
$\frac{1-a}{n}+2k\frac{1-a}{n}\max \{1/l,1/n \}+2\frac{a}{nl}\leq$ \hfill (being very generous)\\
$\frac{1-a}{n}+(2k+2)\frac{1-a}{n}\max \{1/l,1/n \}$

For $i \in Costly-\{i_k\}$, $2.b$ holds since variables $x_{ij}$ were
not changed. \\

{\sf Invariant $3$}

For $i \in Cheap$ we have:\\

\noindent
$\sum_j x'_{ij}=\sum_j x_{ij} + \frac{1}{n}\sum_{j}x_{i_kj} +o(1)\leq$ \hfill (by Invariants $3,$ $4$ of i.h.)\\
$(nU+1)\frac{1-a}{n}+2k(nU+1)\frac{a}{nl} +\frac{(nU+1)\frac{a}{l}+k}{n}+o(1)\leq$\\
$(nU+1)\frac{1-a}{n}+(2k+2)(nU+1)\frac{a}{nl}$\\

The $o(1)$ above is due to the fact that at most $k$ assignment variables
 for some cheap facilities may have been touched as type $2$ and are $0.$ 
For  those same clients the assignment to $i_k$ is fractional,
 so the demand
corresponding to them  that was assigned to  $i_k,$  must be 
distributed among the, at least $n-k$, available
cheap facilities. That additional demand is 
at most $\frac{k (\frac{a}{l}+2k\frac{a(1-a)}{nl})}{n-k} = o(1).$

{\sf Invariant $4$}

For $i \in Costly-{i_k}$ Invariant $4$ holds by inductive hypothesis. For $i_k$ we have $\sum_j x_{i_kj}=0$.\\

\medskip
\noindent
{\bf \underline{subcase $2b$:} touched variable is  $x_{i_kj*}$, $i_k \in Costly$}

{\sc Algorithm}

Consider the type $2$ children $(y',x')$ of $v$ corresponding to variable $x_{i_kj*}$. 
Let $f=\frac{x_{i_kj*}}{1-x_{i_kj*}}$. Solution $(y',x')$ is dictated by its twin type $1$ node-solution (case $1.1b$):
variable $y_{i_k}$ is multiplied by a factor of $(1-f(1/y_{i_k}-1))$ and  for $i\neq i_k$, $x'_{ij*}=x_{ij*}(1+f)$ and
  $x_{i_kj*}=0$. Every other variable remains the same.

{\sc Feasibility}
The feasibility of this case is trivial by Lemma \ref{eqconlemma}.

{\sc Invariants}

{\sf Invariant $1$}

For facilities $i\in Costly-\{i_k\}$ the proof is trivial (no change).
For $i_k$ we have:

\noindent
$y'_{i_k}=y_{i_k}(1-f(1/y_{i_k}-1))=$\\
$y_{i_k}-(1-y_{i_k})\frac{x_{i_kj*}}{1-x_{i_kj*}}\geq$ \hfill (by Invariants $1,$ $2$ of i.h.) \\
$b-2k\frac{a}{l}-2\frac{a}{l}\geq$\\
$b-(2k+2)\frac{a}{l}$\\

{\sf Invariant $2$}

For client $j*$ and facility $i\in Cheap$ we have $2.a$:

\noindent
$x_{ij*}=x_{ij*}(1+f)\leq$ \hfill (by Invariant $2$ of i.h.)\\
$\frac{1-a}{n}+2k\frac{1-a}{n}\max\{1/l,1/n\}+2\frac{(1-a)a}{nl}\leq$\\
$\frac{1-a}{n}+(2k+2)\frac{(1-a)}{n}\max\{1/l,1/n\}$\\

For client $j*$ and facility $i\in Costly$ we have $2.b$:

\noindent
$x_{ij*}=x_{ij*}(1+f)\leq$
$\frac{a}{l}+2k\frac{a(1-a)}{nl}+2\frac{a^2}{l^2}\leq$\\
$\frac{a}{l}+2(k+1)\frac{a(1-a)}{nl}$

Case $2.c$ similarly.

{\sf Invariant $3$}

For $i\in Cheap$:

\noindent
$\sum_j x'_{ij}\leq \sum_j x_{ij} + 1\leq$ \hfill (by $3$ of i.h.)\\
$(nU+1)\frac{1-a}{n}+2(k+1)(nU+1)\frac{a}{nl}$\\

{\sf Invariant $4$}

For $i_k$ the total demand is decreased while for $i\in
Costly-\{i_k\}$:

\noindent
$\sum_j x'_{ij}\leq \sum_j x_{ij} + 1\leq$ \hfill (by $3$ of i.h.)\\
$(nU+1)\frac{a}{l}+k+1$\hfill  if $y_i \neq 1$  or \\
$((nU+1)\frac{a}{l}+k+1)b^{-1}$\hfill if $y_i =1$\\

\medskip 

\noindent
{\bf \underline{subcase $2c$:} touched variable is  $x_{i_kj*}$, $i_k \in Cheap$}

{\sc Algorithm}

Consider the type $2$ children $(y',x')$ of $v$ corresponding to variable $x_{i_kj*}$.
Let $f=\frac{x_{i_kj*}}{1-x_{i_kj*}}$. Solution $(y',x')$ is dictated by its twin type $1$ node-solution (case $1.1c$):
variables $y_{i}\notin \{0,1\}, i \in Costly,$ are multiplied by a
factor of $(1+f\frac{(1/y_{i}-1)x_{ij*}}{x_{i_kj*}t}),$ where 
 $t$ is again the number of
facilities in $Cheap$ for which $j$ is assigned with a non zero fraction (so $t \geq n-k$).
 For $i\neq i_k$, $x'_{ij*}=x_{ij*}(1+f)$ while $x'_{i_kj*}=0$. Every other variable remains the same.

{\sc Feasibility}
The satisfaction of \eqref{eq} is ensured by Lemma \ref{eqconlemma}.

{\sc Invariants}

{\sf Invariant $1$}

For  facility $i \in Costly$ such that $y_i \notin \{0,1\}$ we have:

 \noindent
$y'_{i}=y_{i}(1+f\frac{(1/y_{i}-1)x_{ij*}}{x_{i_kj*}n})=$\\
$y_{i}+(1-y_{i})\frac{x_{ij*}}{(1-x_{i_kj*})t}\leq$\hfill (by Invariants $1,$ $2$ of i.h.) \\
$b+2k\frac{a}{l}+2\frac{a}{nl}\leq$\\
$b+(2k+2)\frac{a}{l}$\\

{\sf Invariant $2$}

For client $j*$ and facility $i\in Cheap$ we have $2.a$:

\noindent
$x_{ij*}=x_{ij*}(1+f)\leq$ \hfill (by Invariant 2 of i.h.)\\
$\frac{1-a}{n}+2k\frac{1-a}{n}\max\{1/l,1/n\}+2\frac{(1-a)^2}{n^2}\leq$\\
$\frac{1-a}{n}+(2k+2)\frac{1-a}{n}\max\{1/l,1/n\}$\\

For client $j*$ and facility $i\in Costly$ we have $2.b$:

\noindent
$x_{ij*}=x_{ij*}(1+f)\leq$\hfill (by Invariant $2$ of i.h.)\\
$\frac{a}{l}+2k\frac{a(1-a)}{nl}+2\frac{(1-a)a}{nl}\leq$\\
$\frac{a}{l}+(2k+2)\frac{a(1-a)}{nl}$\\

{\sf Invariant $3$}

The demand assigned to $i_k$ is decreased. For $i\in Cheap - \{i_k\}$:

\noindent
$\sum_j x'_{ij}\leq \sum_j x_{ij} + 1\leq$\hfill (by $3$ of i.h.)\\
$(nU+1)\frac{1-a}{n}+2(k+1)(nU+1)\frac{a}{nl}$\\

{\sf Invariant $4$}

For $i\in Costly$:

\noindent
$\sum_j x'_{ij}\leq \sum_j x_{ij} + 1\leq$\hfill (by $3$ of i.h.)\\
$(nU+1)\frac{a}{l}+k+1$\hfill if $y_i \neq 1$ or \\
$((nU+1)\frac{a}{l}+k+1)b^{-1}$\hfill if $y_i =1$\\

\medskip \medskip 
The case analysis is complete. 
It remains to show that the witness vectors we constructed for node
$v$ satisfy the symmetry requirements.

\begin{lemma}\label{symmetry}
The symmetry condition, as stated in Claim \ref{simplefact2},
 is satisfied  for  the children of node-solution $v.$
\end{lemma}

\begin{proof}
By construction we never alter integer values of variables, therefore
the condition of Claim~\ref{simplefact1} holds. 

When a variable  $y_i$, $i \in Costly,$ is touched then  for the symmetry between
$y_i$ and each other variable we have:

 For all $j$, variables $x_{ij}$ are multiplied by $1/y_i$ (case $1a$), and when
 some $x_{ij}$ is touched, variable $y_i$ is multiplied by $1/y_i$ (case $1b$).

  For  all  $j$,   variables  $x_{i'j}$,  $i'  \in  Cheap,$   are  multiplied  by
  $(1- (1/y_i-1)\frac{x_{ij}}{x_{i'j}t})$ (case $1a$), and when some $x_{i'j}$ is
  touched, variable $y_i$ is multiplied by $(1- (1/y_i-1)\frac{x_{ij}}{x_{i'j}t})$
  (case $1c$).

 For all $j$, variables $y_{i''}$, $x_{i''j}$, $i'' \in Costly-\{i\},$ are
 multiplied  by  $1$ (case  $1a$),  and when  $y_{i''}$  or  some $x_{i''j}$  is
 touched, variable $y_i$ is multiplied by $1$ (cases $1a$, $1b$). \\\\

When  a variable  $x_{ij},$ $i  \in  Costly,$ is  touched then  for the  symmetry
between $x_{ij}$ and each other variable we have:

For all  $j'\neq j$  and all  $i'$, variables $x_{i'j'}$  are multiplied  by $1$
(case  $1b$),  and  when  some  $x_{i'j'}$  is  touched,  variable  $x_{ij}$  is
multiplied by $1$ (cases $1b$, $1c$).

  For $i'  \neq i$, variables $x_{i'j}$  are multiplied by $0$  (case $1b$), and
  when some $x_{i'j}$ is touched,  variable $x_{ij}$ is multiplied by $0$ (cases
  $1b$, $1c$). \\\\

Finally, when  variable  $x_{ij},$ $i  \in  Cheap,$ is  touched then  for the  symmetry
between $x_{ij}$ and each other variable, the remaining cases that have not been
covered above are: 

For all $j'\neq  j$ and all $i' \in Cheap,$  variables $x_{i'j'}$ are multiplied
by  $1$  (case  1c), and  when  $x_{i'j'}$  is  touched, variable  $x_{ij}$  is
multiplied by $1$ (case $1c$). 

For all  $i' \in  Cheap,$ variables $x_{i'j}$ are multiplied by $0$ (case
$1c$),  and when $x_{i'j}$  is touched,  variable $x_{ij}$  is multiplied  by $0$
(case $1c$).

\end{proof}

The proof of Theorem~\ref{smalltheorem} is now complete.

The proof 
yields a tradeoff between 
the number of rounds as a function of the dimension of the instance  
and the integrality gap, 
which can be obtained by toying with the quantities $U$, $a$, and $b$ that are
left as parameters. One can obtain a higher gap that survives for a
smaller number of rounds.

\section{Discussion}\label{disc}
\label{sec:open}

It is not hard to see that our proof of Theorem~\ref{theorem:ls-cfl} 
also yields the same lower bound 
for the mixed LS$_{+}$ \cite{Cornuejols08}
procedure:  simply restrict  the constructed  protection  matrices  to the  $y$
variables. The  resulting matrices  are of the  form $yy^T +  Diag(y -
y^2)$   which    are   well-known   to    be   positive   semidefinite
(see, e.g., \cite{GoemansT01}).

 The \cfl\ instance for which  the LS procedure fails is essentially a
 Minimum Knapsack instance which can be approximated within a constant
 factor by adding the, exponentially many, knapsack-cover inequalities
 \cite{CarrFLP00}. Note that such  an instance might be a sub-instance
 of a  larger \cfl\ instance  with positive connection costs.   To add
 constraints  of  the knapsack-cover  flavor  would  at least  require
 preprocessing to recognize sub-instances  that are similar to the one
 in our  proof, assuming that  such a task  can be done  in polynomial
 time.  ``Similar''  would mean clusters of closely  located cheap and
 costly facilities  and clients,  where the definition  of "closely'',
 ``cheap'' and ``costly"  would depend somehow on the  actual costs in
 the instance. We would like to emphasize that our proof on the number
 of  rounds  in  Theorem~\ref{smalltheorem}  is  robust  since  it  is
 completely  independent of the  cost structure  of the  instance. One
 could modify all the facility opening and connection costs,
the survival  of the
  fractional solution $(y,x)$ is guaranteed.

 Theorem~\ref{theorem:ls-cfl} implies that the LS lift-and-project method
 fails to capture an efficient strong formulation for \cfl, 
including any useful preprocessing steps as sought by \cite{AnBS13}. 
It would be interesting to complement our  result with a similar one
on the SA hierarchy.

Theorems~\ref{theorem:proper}
and \ref{theorem:proper2} on proper relaxations rule out a constant
integrality gap for ``configuration''-type symmetric LPs of
superpolynomial size, without any  assumptions on the time required to
solve them.  Obtaining a non-symmetric proper LP with a small
gap, if one  exists, seems to require looking  into the cost structure
of the instance, which would entail again some sort of preprocessing.
Of course, one should be careful about calling an algorithm with 
 drastic preprocessing    relaxation-based.

Finally, we conjecture that there is a bad fractional solution for
\lbfl\  that  survives $\omega(1)$  rounds  of  the  LS
procedure.

\bibliographystyle{plain}

\bibliography{bibliography-ver1}

\appendix

\section{Appendix to Section~\ref{sec:prel}}

\begin{lemma} {(\bf Folklore)}  \label{lemma:ap-classic}
Let $I(F,C)$ be an instance of \lbfl\ (\cfl\/) and $z_I^c, z_I^s$ the
corresponding optimal values of 
relaxations (LP-classic) and (LP-star). Then $z_I^c = z_I^s.$
\end{lemma}
\begin{proof}
It is easy to see that for any feasible solution to (LP-star) that
satisfies  $\sum_{s \ni i} x_s \leq 1,$ for all $i \in F,$ we can 
construct a solution to (LP-classic) of the same cost. We set $y_i =
\sum_{s \ni i} x_s,$  and $x_{ij} = \sum_{s \mid \{i,j\} \in s} x_s.$

For the converse, we are given a feasible solution $(y,x)$ to
(LP-classic) and we wish to produce a solution $x'$ to (LP-star) of
the same cost.  We proceed to define the stars in the support of $x'.$ 

Fix a facility $i \in F,$ with $y_i > 0.$ Consider  a rectangle $R_i$ of height 
$y_i$ and width $w_i= \lceil \sum_j x_{ij}/y_i \rceil.$ By the
feasibility of $(y,x),$ $w_i \geq B.$  We consider  the quantity
$\sum_j x_{ij}$ as fractional weight that we will pack 
within $R_i.$ 
We divide the rectangle $R_i$ into $w_i$ vertical strips of width $1$
and height $y_i$ that are initially
empty.  
We
start packing from height $h_1=0.$ Let $1 \leq P \leq w_i$ be the
current strip position. For the current client
$j$ we pack weight within the current strip   starting from the current height
$h_{l-1}$ and we
update $h_l$ to $\min \{ y_i, h_{l-1} + x_{ij} \},$ $l >1.$ If 
$h_l = y_i,$ this means that we can pack  no more weight   at the
current position $P;$ we set $h_{l+1}=0$ and pack the remaining quantity $x_{ij} - (y_i
- h_{l-1})$ in the next strip at position $P+1.$ Because $y_i \geq
x_{ij},$ every client $j$ will be fully packed by using at most 
two 
consecutive strips. By  the definition of $w_i$ we have enough
area to pack all of $\sum_j x_{ij}$ within $R_i.$ 

For every value of $h_l$ that was used by the packing algorithm   draw a horizontal
line that stabs $R_i$ at this height. These lines partition $R_i$  into
regions that are rectangles of width $w_i.$ 
 Each of them intersects at least $\max \{ w_i -1, B
\}$ non-empty vertical strips. Because for every $j,$ $x_{ij} \leq
y_i$ no two of 
these non-empty strips contain fractional weight corresponding to the  same $j \in C.$ The clients
corresponding to  those strips, together with $i$ form a star $s.$ We
set $x_{s}'$ equal to the height of the horizontal region. 
We repeat the process above for every $i \in F.$ 
It is easy
to see that in this way we have produced a solution $x'$ that is
feasible  for (LP-star) and has the  same cost as $(y,x).$ 
In the case of \cfl\ the proof is similar.
\end{proof}

\end{document}